\documentclass[11pt]{article}
\usepackage{booktabs}

\usepackage{amsmath,amssymb,amsthm}
\usepackage[colorlinks=true, citecolor=blue]{hyperref}
\usepackage{physics}
\usepackage{fullpage}
\usepackage{color}
\usepackage{breqn}
\usepackage{kpfonts}
\usepackage{bm}
\usepackage{upgreek}

\newtheorem{definition}{Definition}[section]
\newtheorem{lemma}{Lemma}[section]
\newtheorem{theorem}{Theorem}[section]

\DeclareMathOperator*{\argmin}{argmin}

\DeclareMathOperator*{\E}{\mathbb{E}}
\DeclareMathOperator*{\Err}{Err}

\newcommand{\cmmnt}[1]{\ignorespaces}

\title{Low-depth gradient measurements can improve convergence in variational hybrid quantum-classical algorithms}
\author{Aram Harrow \thanks{Center for Theoretical Physics, MIT. \href{mailto:aram@mit.edu}{aram@mit.edu}} \and John Napp \thanks{Center for Theoretical Physics, MIT. \href{mailto:napp@mit.edu}{napp@mit.edu}}}

\begin{document}


\maketitle 

\begin{abstract}
A broad class of hybrid quantum-classical algorithms known as ``variational algorithms'' have been proposed in the context of quantum simulation, machine learning, and combinatorial optimization as a means of potentially achieving a quantum speedup on a near-term quantum device for a problem of practical interest. Such algorithms use the quantum device only to prepare parameterized quantum states and make simple measurements. A classical controller uses the measurement results to perform an optimization of a classical function induced by a quantum observable which defines the problem. While most prior works have considered optimization strategies based on estimating the objective function and doing a derivative-free or finite-difference-based optimization, some recent proposals involve directly measuring observables corresponding to the gradient of the objective function. The measurement procedure needed requires coherence time barely longer than that needed to prepare a trial state. We prove that strategies based on such gradient measurements can admit substantially faster rates of convergence to the optimum in some contexts. We first introduce a natural black-box setting for variational algorithms which we prove our results with respect to. We define a simple class of problems for which a variational algorithm based on low-depth gradient measurements and stochastic gradient descent converges to the optimum substantially faster than any possible strategy based on estimating the objective function itself, and show that stochastic gradient descent is essentially optimal for this problem. Importing known results from the stochastic optimization literature, we also derive rigorous upper bounds on the cost of variational optimization in a convex region when using gradient measurements in conjunction with certain stochastic gradient descent or stochastic mirror descent algorithms.
\end{abstract}

\section{Introduction}
\subsection{Variational hybrid quantum-classical algorithms}

As quantum computing enters the era of Noisy Intermediate-Scale Quantum (NISQ) technology \cite{preskill2018quantum}, a research effort has developed which aims to understand the capabilities and limitations of quantum computing machines which suffer from small qubit numbers, lack of quantum error-correction, and short coherence times. Given the stringent limitations on the quantum-computational power of such devices, it is natural to look for algorithms which offload as much of the computation as possible to a classical computer. One such class of algorithms, known as \emph{variational hybrid quantum-classical} or just $\emph{variational}$ algorithms, has in recent years been proposed to try to harness some quantum speedup while requiring only very modest quantum resources. The fundamental idea of this class of algorithms is simple. It is assumed that one can prepare states belonging to some parameterized family $\ket{\bm{\uptheta}}$ for $\bm{\uptheta}\in \mathcal{X} \subset \mathbb{R}^p$, where $p$ is the number of variational parameters. The set of parameterized states which may be prepared will depend on the specifications of the quantum device. We will consider parameterizations consisting of $p$ ``pulses'' applied to some easy-to-prepare starting state $\ket{\Psi}$:
\begin{equation*}
\ket{\bm{\uptheta}} := \ket{\bm{\uptheta}_1, \dots, \bm{\uptheta}_p} := e^{-iA_p \bm{\uptheta}_p / 2}\cdots e^{-i A_1 \bm{\uptheta}_1 / 2} \ket{\Psi}
\end{equation*}
where $A_j$ is the Hermitian operator which generates pulse $j$. This is the form of variational state most commonly encountered in the literature on variational algorithms, and is also motivated theoretically \cite{mcclean2016theory,yang2017optimizing,bapat2018bang}. Note that in the presence of noise, the parameterized family of states which may be prepared will actually consist of mixed states. We only consider the noiseless case in this paper. It may also be the case that there are more pulses than independent parameters. For instance, one may impose a constraint like $\bm{\uptheta}_i = \bm{\uptheta}_j$. We assume for simplicity that the parameters are independent, but comment on how our results could be easily extended to this case. 

It is assumed that the quantum device is controlled by a classical ``outer loop", and the quantum device is used only for preparing simple quantum states and making simple measurements. The classical outer loop uses this measurement information to perform a \emph{classical} optimization of some function $f(\bm{\uptheta})$ over the feasible set $\mathcal{X}$, where the objective function $f(\bm{\uptheta})$ is induced by some Hermitian \emph{objective observable} $H$, via the relation $f(\bm{\uptheta}) := \expval{H}{\bm{\uptheta}}$. Algorithms of this family have been proposed in the context of quantum simulation (e.g. variational quantum eigensolvers \cite{peruzzo2014variational,wecker2015progress}), combinatorial optimization (e.g. QAOA \cite{farhi2014quantum}), and machine learning (e.g. quantum classifiers \cite{farhi2018classification,mitarai2018quantum,schuld2018quantum,schuld2018circuit,havlicek2018supervised}).

As a simple example, in a simulation context, $H$ could be some physical Hamiltonian for which we want to approximately obtain the ground state energy. If the true ground state (or a state close to the true ground state) belongs to the parameterized family $\{\ket{\bm{\uptheta}}\}_{\bm{\uptheta}}$ for $\bm{\uptheta}\in \mathcal{X}$, approximately minimizing $f(\bm{\uptheta})$ yields an approximation to the ground state energy.  The typical way a variational algorithm would obtain information about $f(\bm{\uptheta})$ with little quantum resources is by expanding $f(\bm{\uptheta}) = \expval{H}{\bm{\uptheta}} = \sum_{i=1}^m \alpha_i \expval{P_i}{\bm{\uptheta}}$ where $P_i$ are tensor products of Pauli operators (recall that any operator may be expanded in this way), or are some other observables which may be easily measured. In this paper, we assume the $P_i$ are products of Pauli operators. Assuming it is possible to easily measure these operators, one can estimate $f(\bm{\uptheta})$ by estimating each $\expval{P_i}{\bm{\uptheta}}$ separately and combining the results according to the coefficients $\alpha_i$. Of course, due to the randomness of the measurement outcomes, many preparations of $\ket{\bm{\uptheta}}$ and measurements may be required to obtain a good estimate of $f(\bm{\uptheta})$. 

This effect occurs generally in variational hybrid algorithms: the randomness of the quantum measurement outcomes translates into the classical outer loop only having \emph{stochastic} access to the objective function $f$. That is, at point $\bm{\uptheta}$ in parameter space, it cannot directly observe the function value $f(\bm{\uptheta})$, but rather some random variable whose expectation value is $f(\bm{\uptheta})$. Numerical simulations that overlook this fact may be misleading. For instance, letting $\epsilon$ denote the optimization error, some problems that admit $\log(1/\epsilon)$ convergence rates given noiseless access to the objective function values admit $\text{poly}(1/\epsilon)$ convergence rates in the stochastic setting \cite{bubeck2015convex}. 

But even worse for prospects of optimization, the resulting classical stochastic optimization problem will generally be complicated and nonconvex, and hence may be intractable. However, one can hope for heuristics which find a reasonable approximate solution. Furthermore, if the algorithm is in a convex vicinity of an optimum, or local optimum, algorithms like stochastic gradient descent which are known to converge for convex problems may converge to the local optimum despite the problem being globally nonconvex. For some proposed applications of variational algorithms one desires a very precise solution, so it is likely that much time will be spent converging in the vicinity of an optimum and this situation may be especially relevant. In this paper, we focus on this latter scenario of convergence within a convex region containing a local optimum, either assuming or proving that this case applies.

\subsubsection{Analytic gradient measurements}

In the usual formulation of variational algorithms, the classical outer loop is assumed to take some number of quantum measurements at a point $\bm{\uptheta}$ in parameter space in order to approximate the objective function value $f(\bm{\uptheta})$ at that point, and then perform an optimization based on these values. However, one can imagine more complicated algorithms which, instead of taking measurements to estimate $f(\bm{\uptheta})$ at point $\bm{\uptheta}$, take measurements corresponding to some other property of the optimization problem. Indeed, a natural alternative choice is to take measurements corresponding to $\nabla f (\bm{\uptheta})$, the \emph{gradient} of the objective function at point $\bm{\uptheta}$. Such a strategy was proposed in the context of combinatorial optimization \cite{guerreschi2017practical}, quantum chemistry \cite{romero2018strategies}, and machine learning \cite{mitarai2018quantum, schuld2018circuit, farhi2018classification}. Similar ideas were also proposed in the context of implementing Hamiltonian evolution in a low-depth variational setting \cite{li2017efficient}. Low-depth procedures for directly measurement gradients in variational algorithms typically require only marginally greater quantum circuit depth than that required for measuring the objective function. Some such methods are reviewed and extended in \cite{bergholm2018pennylane,schuld2018evaluating}.

However, it was not clear that such gradient-measurement based strategies could confer any advantage over objective-measurement based strategies. For one, note that it is possible to obtain an estimate of $\nabla f(\bm{\uptheta})$ using only estimates of the objective function $f(\bm{\uptheta})$. To see how, note that for small $\epsilon$,
\begin{equation*}
\pdv{f}{\bm{\uptheta}_i} (\bm{\uptheta}) \approx \frac{1}{2\epsilon} \qty( f(\bm{\uptheta} + \epsilon \hat{e}_i) - f(\bm{\uptheta} - \epsilon \hat{e}_i) )
\end{equation*}
where $\hat{e}_i$ is the unit vector along the $i\textsuperscript{th}$ component. Estimating the gradient in this way is called \emph{finite-differencing}. A variational algorithm could estimate the gradient at some point while only taking objective function measurements via finite-differencing. To distinguish between strategies which estimate the gradient based on finite-differencing and those which directly take measurements corresponding to the gradient, the measurements of the latter strategy are sometimes referred to as \emph{analytic gradient measurements}. Previously, it could not be ruled out that variational algorithms based on finite-differencing or some other strategy that only requires objective function estimates could always achieve similar performance to strategies based on gradient measurements. Indeed, one paper \cite{guerreschi2017practical} numerically found that  a gradient measurement based strategy performed no better for a certain set of combinatorial optimization problems than a strategy based on finite-differencing.  In this paper, we settle the question by proving that in some contexts, gradient measurements can lead to substantially faster convergence rates than those of all algorithms which only take measurements corresponding to the objective function $f(\bm{\uptheta})$.

For the remainder of this paper, we use the term \emph{zeroth-order} to refer to variational algorithms which only take measurements corresponding to the objective function to be minimized. This class includes algorithms which estimate the gradient by finite-differencing. Similarly, the term \emph{first-order} refers to algorithms which may also take analytic gradient measurements. Generalizing further, a variational algorithm which makes analytic {$k\textsuperscript{th}$-order derivative measurements is a \emph{$k\textsuperscript{th}$-order} algorithm. We make these terms precise in Section \ref{sec:black-box}.

\subsection{Summary of results}\label{sec:SummaryOfResults}

\subsubsection*{Black-box formulation}
In order to rigorously prove a lower bound on the number of quantum measurements required for zeroth-order variational optimization, it is convenient to introduce a black-box formalism. To motivate the introduction of a black-box formalism, note that any variational algorithm can be simulated by a purely classical algorithm that makes zero quantum measurements. Of course, the time complexity of the classical simulation may be exponential in the problem size. 

One encounters a similar problem in the classical setting of trying to quantify the complexity of optimization. The objective function to be minimized could be extremely complicated and difficult to study analytically (for instance, it could correspond to the output of some complicated algorithm) or otherwise inaccessible. A black-box model was therefore developed for the study of convex optimization \cite{nemirovsky1983problem} and remains popular in current research in the field. In this setting, the  function to be optimized is encoded in an oracle, and we define the query complexity of an algorithm for optimizing the function to be the number of calls made to the oracle. The algorithm may be promised that the objective function has certain properties, but is not given an exact description of the function. This black-box formalism provides a natural and general setting for proving bounds in convex optimization.

Similarly, we introduce a black-box model for variational algorithms. In our black-box model, the classical outer loop is not given a full description of the objective observable $H$, but is rather given an oracle $\mathcal{O}_H$ encoding $H$. It could also be promised that $H$ has a certain structure, but is not given an exact description of $H$. The outer loop may query $\mathcal{O}_H$ with a specification of a $p$-pulse state parameterization $\Theta$, a parameter $\bm{\uptheta}\in \mathbb{R}^p$, and a multiset $S$ containing integers in $\{ 1,\dots,p\}$. Given this input, the black box internally prepares the state $\ket{\bm{\uptheta}}$ and performs a simple randomized procedure which involves making a measurement of a single tensor product of Pauli operators, and then outputs a random variable $X$ such that $\E X = f(\bm{\uptheta})$ if $S = \varnothing$  (a zeroth-order query) and $\E X = \frac{\partial^k f}{\partial \bm{\uptheta}_{s_1}\partial \bm{\uptheta}_{s_2}\cdots \partial \bm{\uptheta}_{s_k}}(\bm{\uptheta})$ if $S = \{ s_1, \dots, s_k \}$ (a $k\textsuperscript{th}$-order query). This procedure for obtaining unbiased estimates of derivatives can be implemented in low depth in practice, assuming the parameterized states $\ket{\bm{\uptheta}}$ can be prepared in low depth. Essentially, our black-box model puts the entire ``quantum'' part of the variational algorithm into the black box. The internal randomized procedure that the black box runs to choose an observable to measure uses a natural importance sampling strategy, whereby terms of the objective observable with smaller norm are sampled with smaller probability. The desirable consequence of this strategy is that terms of small norm contribute little to the variance of the output of the black box.

It should be noted that, when variational algorithms are used in practice, commuting terms in the Pauli decomposition of $H$ can be measured on a single trial state. To simplify the analysis, our black-box model does not take advantage of this possible speedup. Hence, the query cost in our model corresponds to the number of products of Pauli operators measured, rather than the number of state preparations required, which is a related quantity but could be lower. (We comment on how taking this into account would affect our bounds for the toy model we study in Table \hyperlink{tab2}{2}.)

\subsubsection*{General query complexity upper bounds for stochastic gradient descent, stochastic mirror descent, and zeroth-order strategies in a convex region}

\begin{table}
\centering
\begin{tabular}{llll}
\toprule
Convexity of $f(\bm{\uptheta})$ & Zeroth-order & SGD & SMD \\
\midrule
Convex & $\min\qty(\frac{p^{32} E^2}{\epsilon^2}, \frac{p^2 E^4 (R_2/r_2)^2}{\epsilon^4})$ & $\frac{R_2^2 \|\vec{\Gamma}\|_1^2}{\epsilon^2}$ & $\frac{R_1^2 \|\vec{\Gamma}\|_2^2 }{\epsilon^2}$ \\

$\lambda_2$-strongly convex w.r.t. $\|\cdot \|_2$ & $\min\qty(\frac{p^{32} E^2}{\epsilon^2}, \frac{p^2 E^4 (R_2/r_2)^2}{\epsilon^4})$ & $\frac{ \|\vec{\Gamma}\|_1^2}{\lambda_2 \epsilon}$ & $ \frac{p\|\vec{\Gamma}\|_2^2 }{\lambda_2 \epsilon}$ \\

$\lambda_1$-strongly convex w.r.t. $\| \cdot \|_1$  & $\min\qty(\frac{p^{32} E^2}{\epsilon^2}, \frac{p^2 E^4 (R_2/r_2)^2}{\epsilon^4})$ & $ \frac{ \|\vec{\Gamma}\|_1^2}{\lambda_1 \epsilon}$ & $\frac{\|\vec{\Gamma}\|_2^2}{\lambda_1 \epsilon}$ \\
\bottomrule
\end{tabular}
\caption[LoF entry]{Rigorous upper bounds for the query complexity of optimizing $f(\bm{\uptheta})$ to precision $\epsilon$ in a convex region $\mathcal{X} \subset \mathbb{R}^p$ contained in a $2$-ball of radius $R_2$, contained in an $1$-ball of radius $R_1$, and containing a $2$-ball of radius $r_2$, using zeroth-order strategies or analytic-gradient measurements in conjunction with stochastic gradient descent (SGD) or stochastic mirror descent (SMD) with an $l_1$ setup.  $E$ and $\vec{\Gamma}$ are parameters related to the coefficients of the Pauli expansions of the objective observable $H$ and pulse generators $A_j$, and are defined in Section \ref{sec:black-box} and summarized in Table \ref{tab:notation}. Constants, logarithmic factors, and some Lipschitz constants of the objective function are hidden.

\hspace{0.5cm} For the toy family of objective observables on $n$ qubits $\mathcal{H}_n^\epsilon$ that we analyze to prove our separation in Section \ref{sec:Separation}, and with respect to the associated  ansatz and feasible set defined in Section \ref{sec:upperBoundProof}, we have $p = n$, $E = \Theta(n)$, $\Gamma_i = \Theta(1)$,  $R_2 = \Theta(\sqrt{\epsilon})$, $R_1 = \Theta(\sqrt{\epsilon n})$, $r_2 = O(\sqrt{\epsilon/n})$, $\lambda_2 = \Theta(1)$, and $\lambda_1 = \Theta(1/n)$. In this case, SGD and SMD have the same asymptotic performance up to polylogarithmic factors.}\label{tab:introUpper}
\end{table}

We consider a restriction of the variational problem to a convex region of parameter space $\mathcal{X}$ on which the objective function $f(\bm{\uptheta})$ is assumed to be convex. We import known results from the stochastic optimization literature to obtain convergence rates for various optimization strategies and for various assumptions about $f(\bm{\uptheta})$, such as strong convexity. We derive upper bounds on the query cost for strategies that use analytic gradient measurements in conjunction with stochastic gradient descent or stochastic mirror descent with an $l_1$ setup. The upper bounds are functions of the dimension of parameter space $p$, precision $\epsilon$, geometry of the feasible set, and coefficients in the Pauli expansions of the objective observable and pulse generators.  

Stochastic mirror decent (see Appendix \ref{sec:mirrorDescent}, or e.g. \cite{nemirovski2009robust,juditsky2011first,bubeck2015convex} for more detailed reviews), which we will abbreviate as SMD, can be thought of as a generalization of SGD to non-Euclidean spaces. As motivation for why SMD might be relevant in our setting, note that the 1-norm of a parameter vector has a very natural interpretation. Namely, since $e^{-i A_j \bm{\uptheta}_j/2}$ is essentially the unitary evolution generated by $A_j / 2$ for time $\bm{\uptheta}_j$, $\| \bm{\uptheta} \|_1 := \sum_i |\bm{\uptheta}_i|$ may be interpreted as the total amount of time that the starting state $\ket{\Psi}$ is evolved for to reach the trial state $\ket{\bm{\uptheta}} := e^{-i A_p \bm{\uptheta}_p / 2} \cdots e^{-i A_1 \bm{\uptheta}_1 / 2}\ket{\Psi}$, and $\| \bm{\uptheta}_i - \bm{\uptheta}_j \|_1$ is associated with the amount of time for which the associated pulse sequences differ. On the other hand, SGD is appropriate for Euclidean geometriesa. Taking our norm to be $\|\cdot \|_1$ instead of $\| \cdot \|_2$ and using a suitable version of SMD yields nearly  quadratically better scaling with respect to the dimension of parameter space $p$ in some settings, as compared with SGD. In other settings, SGD outperforms SMD. We record upper bounds based on both strategies. It is an open problem to understand what sort of values the parameters in the upper bounds typically take in practice, and relatedly, whether a Euclidean geometry or some other geometry is more appropriate.

For comparison, we also record a rigorous upper bound for zeroth-order strategies by applying two of the best known zeroth-order upper bounds \cite{flaxman2005online, agarwal2011stochastic} from the stochastic optimization literature. Our results on general upper bounds are displayed in Table \ref{tab:introUpper}. It should be noted that these zeroth-order upper bounds are the best \emph{rigorous} zeroth-order upper bounds that we are aware of, but it is likely that other derivative-free algorithms would significantly outperform these bounds in practice in many instances. For instance, methods based on trust regions and surrogate models \cite{conn2009introduction} often perform very well in practice, despite not necessarily having strong theoretical guarantees. 

This point is an example of a more general limitation that applies to all of the upper bounds we report, relating to the difference between theoretical and empirical results. For one, these are upper bounds for convergence that apply when the algorithm has a trusted region which it knows contains a local optimum, and which the objective function is convex over. Furthermore, while strong convexity of the objective function in the domain can greatly improve the convergence, the rigorous upper bounds that exploit this property apply when the algorithm has a good estimate of this strong convexity parameter. In other language, the upper bounds apply in a promise setting, where the algorithm is promised that a certain convex subset of parameter space contains the optimum, the objective function is convex in this domain, and any other relevant parameters take certain values. However, this may very well not be the case in practice. Furthermore, these are worst-case upper bounds, and may be outperformed in practice.

All of these challenges are well-known in machine learning tasks such as deep learning, in which one may want to converge to a local optimum of some complicated nonconvex optimization problem. In practice, one can do hyperparameter optimization, which would involve yet another classical outer loop varying over the parameters used to define the optimization algorithm itself. Parameters could also be set adaptively (e.g. \cite{kingma2014adam}). Another possibility is that the algorithm could try to construct a surrogate model for the objective function which is valid in some region, and estimate relevant parameters from the surrogate model. Unfortunately, such methods are often not backed up by strong theoretical results, even when the practical performance is very good. While there may be a sizable gap between theory and practice, we hope that the theoretical convergence upper bounds will nonetheless be useful for guiding practical implementations and expectations. 

Finally, we point out that these upper bounds do not explicitly depend on the number of terms in the Pauli expansion of the objective observable or pulse generators, but rather on sums of coefficients in the expansion. This is a desirable feature for applications such as quantum chemistry, where it is often the case that the electronic structure Hamiltonian is written as a sum of a very large number of terms, many of which have very small norm. 

\subsubsection*{Query complexity separation between zeroth-order and first-order variational algorithms for a simple class of objective observables, and optimality of stochastic gradient descent}

\begin{table}
\centering
\begin{tabular}{lll}
\toprule
Sampling oracle available & Lower bound & Upper bound \\
\midrule
Zeroth-order & $\Omega(n^3 / \epsilon^2)$ & $\min\qty( O(n^7 / \epsilon^4), \widetilde{O}(n^{34} / \epsilon^2))$ \\
First-order & $\Omega(n^2 / \epsilon)$ & $O(n^2 / \epsilon)$ \\
All orders & $\Omega(n^2/\epsilon)$ & $O(n^2 / \epsilon)$ \\
All orders \& unrestricted domain & $\Omega(n^2/\epsilon)$ & $O(n^2 / \epsilon)$ \\
\bottomrule
\end{tabular}
\caption{Lower and upper query complexity bounds for optimizing the family $\mathcal{H}_n^\epsilon$ of objective observables on $n$ qubits to precision $\epsilon$ within the vicinity of the optima of $\mathcal{H}^\epsilon_n$. Note that the family $\mathcal{H}_n^\epsilon$ has a structure which allows $n$ terms to be measured with a single state preparation. This fact is not taken into account in our black-box model, which quantifies measurement cost instead of state preparation cost. As such, corresponding upper bounds on the number of state preparations required would be a factor of $n$ smaller than those in the table. The zeroth-order upper bounds are based on the algorithms of \cite{flaxman2005online} and \cite{agarwal2011stochastic}, respectively.}\label{tab:introHn}
\end{table}

After establishing an oracular setting and recording some upper bounds, for a given parameter $\epsilon > 0$ we define a certain class $\mathcal{H}_n^\epsilon$ of simple 1-local objective observables on $n$ qubits which we use to demonstrate a separation in query complexity between variational algorithms which only make zeroth-order  queries to the oracle and those which make first-order queries. The optima of the observables in $\mathcal{H}^\epsilon_n$ are $O(\epsilon)$-close to each other in objective function value, in the sense that for all $H, H^\prime \in \mathcal{H}^\epsilon_n$, if $\ket{\psi}$ is an optimum (i.e. ground state) of $H$, then $\expval{H^\prime}{\psi} - \lambda_{\min}(H^\prime) \leq O(\epsilon)$, where $\lambda_{\min}(H^\prime)$ is the smallest eigenvalue of $H^\prime$.

We show that for any precision parameter $0< \epsilon < \Theta(n)$, any zeroth-order variational algorithm which optimizes any observable $H \in \mathcal{H}^\epsilon_n$ to precision $\epsilon$ and only queries the oracle with states in the vicinity of the optimum of $H$ must make at least $\Omega(n^3/\epsilon^2)$ queries.  Here, the ``vicinity of the optimum'' is essentially the set of states that are $O(\epsilon)$-close to optimal in objective function value. On the other hand, we show that after making  a good choice of variational ansatz, an SGD algorithm that performs analytic gradient measurements to get gradient estimates optimizes any objective observable in the family to expected precision $\epsilon$ with only $O(n^2 / \epsilon)$ queries of states in the vicinity of the optimum.

\begin{theorem}[Informal]
Suppose $\mathcal{A}$ is a variational algorithm that only makes zeroth-order queries to $\mathcal{O}_H$, only queries states in the vicinity of the optima of $\mathcal{H}^\epsilon_n$, and for any $H\in \mathcal{H}^\epsilon_n$ that is realized, outputs a description of a state whose expected objective function value is $\epsilon$-close to $\lambda_{\min}(H)$. Then $\mathcal{A}$ must make $\Omega\qty(\frac{n^3}{\epsilon^2})$ queries.
\end{theorem}

\begin{theorem}[Informal]
There exists a variational algorithm that only makes first-order queries to $\mathcal{O}_H$, only queries states in the vicinity of the optima of $\mathcal{H}^\epsilon_n$, and for any $H\in \mathcal{H}^\epsilon_n$, makes $O\qty(\frac{n^2}{\epsilon})$ queries and outputs a description of a state whose expected objective function value is $\epsilon$-close to $\lambda_{\min}(H)$. An algorithm that achieves this rate is a simple stochastic gradient descent strategy.
\end{theorem}

We also show that, up to a possible constant factor, the strategy of using analytic gradient measurements in conjunction with stochastic gradient descent is in fact optimal for this problem within our black-box setup (SMD with an $l_1$-geometry setup converges at a rate that is only a factor of $O(\log n )$ worse than that of SGD). We do this by proving that even if the algorithm is allowed to make $k\textsuperscript{th}$ order queries for any $k$, an $\Omega(n^2 / \epsilon)$ lower bound still applies. (This bound holds even if the algorithm is allowed to query the oracle with states that are outside the vicinity of the optima of $\mathcal{H}^\epsilon_n$.) This lower bound is matched (up to a constant) by the upper bound of SGD. An interesting consequence of this fact is that, for this particular problem, first-order queries are better than zeroth-order queries, but $k\textsuperscript{th}$ order queries provide no significant benefit over first-order queries. 

\begin{theorem}[Informal]
Suppose $\mathcal{A}$ is a variational algorithm that may make queries of any order to $\mathcal{O}_H$, may query the oracle with any state, and for any $H\in \mathcal{H}^\epsilon_n$ outputs a description of a state whose expected objective function value is $\epsilon$-close to the optimum. Then $\mathcal{A}$ makes at least $\Omega\qty(\frac{n^2}{\epsilon})$ queries.
\end{theorem}

 We summarize the above results in Table \ref{tab:introHn}, where for comparison we also include a rigorous  zeroth-order upper bound based on the results of \cite{flaxman2005online} and \cite{agarwal2011stochastic}.

\subsection{Related work}
In this paper, we are primarily interested in the low-depth setting. The methods we consider for measuring the gradient yield an unbiased, but possibly very noisy estimate of the gradient in low depth. An alternative approach for measuring the gradient in variational algorithms was recently proposed in \cite{gilyen2017optimizing}, which builds on Jordan's gradient measurement algorithm \cite{jordan2005fast}. Their algorithm offers significantly better performance for obtaining precise estimates of the gradient, but also requires significantly more quantum resources, with coherence time requirements increasing with the desired precision. 

A lower bound for a class of derivative-free stochastic convex optimization problems was shown in \cite{jamieson2012query}. Our separation result in Section \ref{sec:Separation} is similar in spirit to their result, and the proof strategies share similarities. However, our setting of variational hybrid algorithms is very different from theirs, preventing their result from being  ported to variational quantum algorithms. In particular, we are interested specifically in stochastic optimization problems induced by a quantum observable and variational ansatz. Furthermore, in our setting the classical outer loop's optimization problem is not fixed; different variational ans\"{a}tze induce different optimization problems. Our lower bound for zeroth-order algorithms takes this extra freedom  into account, applying for \emph{any} choice of variational ansatz  (and also allowing the algorithm to change ansatz over the course of the optimization). Our proof strategy for the zeroth-order lower bound also borrows some techniques from \cite{agarwal2009information}, which showed lower bounds for classes of first-order stochastic optimization problems. In turn, these techniques are inspired by methods in statistical minimax and learning theory.

\subsection{Organization}
In Section \ref{sec:Preliminaries}, we state the conventions we adhere to and record known results from stochastic convex optimization that we later use. In Section \ref{sec:black-box}, we introduce a black-box setting for variational algorithms and define the oracle $\mathcal{O}_H$ encoding an objective observable $H$. In Section \ref{sec:General}, we present general upper bounds on the query cost of optimization for different algorithms and different assumptions on the objective function. In Section \ref{sec:Separation}, we define a parameterized class of variational optimization problems on $n$ qubits $\mathcal{H}_n^\epsilon$ and use this class of problems to prove a query complexity separation between zeroth-order and first-order optimization strategies. We conclude and mention some open questions in Section \ref{sec:conclusion}. Appendix \ref{appendix} contains some relevant background on first-order stochastic convex optimization algorithms.

\section{Preliminaries}\label{sec:Preliminaries}
\begin{table}
\centering
\begin{tabular}{ll}
\toprule
Notation & Meaning \\
\midrule
$H$ & Objective observable. \\
$\lambda_{\text{min}}(H)$ & Smallest eigenvalue of $H$. \\
$\Theta$ & State parameterization of form $e^{-iA_p \bm{\uptheta}_p / 2}\cdots e^{-iA_1 \bm{\uptheta}_1 / 2}\ket{\Psi}$. \\
$p$ & Number of parameters/pulses. Dimension of the optimization problem. \\
$A_j$ & Generator of pulse $j$. \\
$\ket{\Psi}$ & Starting state. \\
$\bm{\uptheta}\in \mathbb{R}^p$ & Parameter. \\
$\ket{\bm{\uptheta}}$ & State corresponding to parameter $\bm{\uptheta}$ (and implicit parameterization $\Theta$). \\
$\mathcal{X} \subset \mathbb{R}^p$ & Feasible set. \\
$R_1, R_2$ & Smallest radius of a $1$-ball or Euclidean ball, respectively, containing $\mathcal{X}$. \\
$r_2$ & Largest radius of a Euclidean ball contained in $\mathcal{X}$. \\
$f(\bm{\uptheta})$ & Induced objective function. Equal to $\expval{H}{\bm{\uptheta}}$. \\
$\alpha_i$, $m$, $P_i$ & $H = \sum_{i=1}^m \alpha_i P_i$ where the r.h.s. is the Pauli decomposition of $H$, and $\alpha_i > 0$.  \\
$\beta^{(j)}_{k}$, $n_j$, $Q^{(j)}_k$ & $A_j = \sum_{k=1}^{n_j} \beta^{(j)}_k Q^{(j)}_k$ where the r.h.s. is the Pauli decomposition of $A_j$, and $\beta^{(j)}_k > 0$. \\
$E$ & $\sum_{i=1}^m \alpha_i$ . Upper bounds the operator norm of $H$. \\
$\gamma^{(j)}_{kl}$ & $0$ or $\beta^{(j)}_k \alpha_l$ (see Section \ref{sec:first-order sampling}). \\
$\Gamma_j$ & $\sum_{k=1}^{n_j} \sum_{l=1}^m \gamma^{(j)}_{kl}$.\\
$\vec{\Gamma}$ & $(\Gamma_1, \Gamma_2, \dots, \Gamma_p)^{\top}$. \\
$\lambda_1, \lambda_2$ & Strong convexity parameter w.r.t. 1-norm or 2-norm, respectively. \\
$\bm{\uptheta}^*$ & Minimizer of $f(\bm{\uptheta})$ on the feasible set. \\
$\vec{r}_i$ & Polarization (Bloch vector) of the reduced state on qubit $i$. \\

\bottomrule
\end{tabular}
\caption{Notation and parameters.}\label{tab:notation}
\end{table}
\subsection{Conventions, assumptions, and notation}
We will assume throughout that variational states are parameterized according to an ansatz of the form
\begin{equation*}
\ket{\bm{\uptheta}} := e^{-i A_p \bm{\uptheta}_p / 2}\cdots e^{-iA_1 \bm{\uptheta}_1 / 2} \ket{\Psi},
\end{equation*}
where $\ket{\Psi}$ is assumed to be some easy-to-prepare starting state, $\bm{\uptheta} := (\bm{\uptheta}_1, \dots, \bm{\uptheta}_p)^{\top} \in \mathcal{X} \subset \mathbb{R}^p$, and $A_i$ are Hermitian operators. We will refer to $\mathcal{X}$ as the \emph{feasible set}. We refer to an individual factor $e^{-i A_j \bm{\uptheta}_j / 2}$ in the ansatz as a \emph{pulse}, and an $A_i$ as a \emph{pulse generator}. Many parameterizations for variational algorithms found in the literature are of this form, and furthermore the papers \cite{mcclean2016theory,yang2017optimizing,bapat2018bang} give evidence supporting this ansatz. We will occasionally need to refer to a specific variational parameterization, often labeled by the character $\Theta$. When we refer to a parameterization $\Theta$, we assume that $\Theta$ collects information about the starting state $\ket{\Psi}$ and the pulse generators $A_i$. 

Given a parameterization $\Theta$ and a feasible set $\mathcal{X}$, the classical objective function to be minimized is induced by some Hermitian operator $H$, which we refer to as the \emph{objective observable}. In particular, the classical objective function is given by
\begin{equation*}
f(\bm{\uptheta}) = \expval{H}{\bm{\uptheta}}, \, \, \bm{\uptheta}\in \mathcal{X}.
\end{equation*}

In the context of variational algorithms, it is assumed that the quantum device is capable of  measuring some subset of quantum observables. We assume that the set of observables which may be measured is the set of all Pauli operators. 

When we refer to a qubit with polarization $\vec{r}$, we mean the state specified by the density matrix $\frac{1}{2}(I + \vec{r}\cdot \vec{\sigma})$, where $\vec{\sigma} := (X,Y,Z)$ is the vector of Pauli operators.  For some vector $\vb{x}\in \mathbb{R}^p$, $\vb{x}_i$ denotes the $i\textsuperscript{th}$ component of $\vb{x}$. Vectors should be considered column vectors by default. Logarithms are assumed to be base 2 unless otherwise specified. The notation $[p]$ for $p\in \mathbb{Z}_+$ denotes the set $\{1,2, \dots, p\}$. The $q$-norm of a vector $\vb{x}\in \mathbb{R}^p$ for $q \geq 1$ is defined as $\| \vb{x} \|_q := (|\vb{x}_1|^q + \dots + |\vb{x}_p|^q)^{1/q}$. The $\infty$-norm is defined as $\| \vb{x}\|_\infty := \max\{|\vb{x}_1|, \dots, |\vb{x}_p|\}$. If $\| \cdot \|$ is an arbitrary norm, the dual norm $\| \cdot \|_*$ is defined as $\| \vb{g} \|_* := \sup_{\vb{x}\, :\, \|\vb{x}\|=1} \vb{g}^{\top} \vb{x}$. The notation $\widetilde{O}(\cdot )$ hides log factors. The notation $\E $ denotes an expectation value. $\hat{e}_j$ denotes the unit vector along coordinate $j$. We let $\lambda_{\text{min}}(H)$ denote the smallest eigenvalue of Hermitian matrix $H$.

We collect notation and parameters in Table \ref{tab:notation}.

\subsection{Requisite results about stochastic convex optimization}\label{sec:Requisite upper bounds}
We will obtain upper bounds for variational algorithms in convex regions by combining well known classical convergence results with sampling strategies for estimating the gradient. Here, we record the classical optimization results we will need. Background on stochastic gradient descent and stochastic mirror descent may be found in Appendix \ref{appendix} (see e.g. \cite{nemirovski2009robust,juditsky2011first,bubeck2015convex} for more thorough reviews). First, we define strong convexity.

\begin{definition}[Strong convexity]
For $\lambda > 0$, the real-valued function $f$ is $\lambda$-strongly convex with respect to norm $\| \cdot \|$  on some convex domain $\mathcal{X}$  if $\forall \vb{x},\vb{y}\in \mathcal{X}$,
\begin{equation*}
f(\vb{y}) \geq f(\vb{x}) + \nabla f(\vb{x})^{\top}(\vb{y}-\vb{x}) + \frac{\lambda}{2}\| \vb{x} - \vb{y}\|^2.
\end{equation*}
\end{definition}
Note that a twice-differentiable function is $\lambda$-strongly convex with respect to the $2$-norm if all of the eigenvalues of the Hessian matrix at each point in the domain are at least $\lambda$. More generally, $f$ is $\lambda$-strongly convex at $\vb{\vb{x}}$ w.r.t. an arbitrary norm $\|\cdot \|$ if $\vb{h}^\top \nabla^2 f(\vb{x}) \vb{h} \geq \lambda \| \vb{h}\|^2$ for all $\vb{h}$, where $ \nabla^2 f(\vb{x})$ is the Hessian of $f$ at $\vb{x}$. In contrast, $f$ is convex if the Hessians are merely positive semidefinite. Intuitively, if $f$ is strongly convex, then it is lower bounded by a quadratic function. Strong convexity can often be used to accelerate optimization \cite{bubeck2015convex}.

\subsubsection{Upper bounds for stochastic first-order optimization}\label{sec:First-order results}
In this section, we record known upper bounds for optimizing convex functions given access to noisy, unbiased gradient information. For the first two results below, we follow the presentation of the review on algorithms for convex optimization \cite{bubeck2015convex}. 

Assume we have access to a stochastic gradient oracle, which upon input of $\vb{x}\in \mathcal{X}$, returns a random vector $\hat{\vb{g}}(\vb{x})$ such that $\E \hat{\vb{g}}(\vb{x}) = \nabla f(\vb{x})$, $\E \| \hat{\vb{g}}(\vb{x}) \|_2^2 \leq G_2^2$, and $\E \| \hat{\vb{g}}(\vb{x}) \|_\infty^2 \leq G_\infty^2$. Assume $\mathcal{X} \subset \mathbb{R}^p$ is a closed convex set. Let $\vb{x}^*$ denote a minimizer of $f$ on $\mathcal{X}$. 

\begin{theorem}[SGD]\label{thm:SGD}
Assume $\mathcal{X}$ is contained in a Euclidean ball of radius $R_2$ and $f$ is convex on $\mathcal{X}$. Then projected SGD with fixed step size $\eta = \frac{R_2}{G_2}\sqrt{\frac{2}{T}}$ satisfies 
\begin{equation*}
\E f\qty(\frac{1}{T}\sum_{s=1}^{T} \vb{x}_s) - f(\vb{x}^*) \leq R_2 G_2 \sqrt{\frac{2}{T}}.
\end{equation*}
where $\vb{x}_1$ is the starting point, and the algorithm visits points $\vb{x}_1, \dots, \vb{x}_T$.
\end{theorem}

\begin{theorem}[SGD for strongly convex functions]\label{thm:SGDSC}
Assume $f$ is $\lambda_2$-strongly convex on $\mathcal{X}$ with respect to $\| \cdot \|_2$. Then SGD with step size $\eta_s = \frac{2}{\lambda_2(s+1)}$ at iteration $s$ satisfies 
\begin{equation*}
\E f\qty( \sum_{s=1}^{T} \frac{2s}{T(T+1)} \vb{x}_s) - f(\vb{x}^*) \leq \frac{2G_2^2}{\lambda_2(T+1)}.
\end{equation*}
where $\vb{x}_1$ is the starting point, and the algorithm visits points $\vb{x}_1, \dots, \vb{x}_T$.
\end{theorem}

\begin{theorem}[SMD with $l_1$ setup \cite{nemirovski2009robust} ]\label{thm:SMD}
Assume $\mathcal{X}$ is contained in a 1-ball of radius $R_1$, and $f$ is convex on $\mathcal{X}$. Then stochastic mirror descent with an appropriate $l_1$ setup and and step size $\eta = \frac{R_1}{G_\infty}\sqrt{\frac{2}{T}}$, satisfies
\begin{equation*}
\E f\qty(\frac{1}{T}\sum_{s=1}^{T} \vb{x}_s) - f(\vb{x}^*) \leq R_1 G_\infty \sqrt{\frac{2e \ln p}{T}}.
\end{equation*}
where $\vb{x}_1$ is the starting point, and the algorithm queries points $\vb{x}_1, \dots, \vb{x}_T$.
\end{theorem}

\begin{theorem}[SMD with $l_1$ setup for strongly convex functions \cite{hazan2014beyond}]\label{thm:SMDSC}
Assume $f$ is $\lambda_1$-strongly convex on $\mathcal{X}$ with respect to norm $\|\cdot \|_1$. Then a certain SMD-like algorithm, running for $T$ iterations,  outputs a (random) vector $\bar{\vb{x}}$ such that
\begin{equation*}
\E f(\bar{\vb{x}}) - f(\vb{x}^*)  \leq \frac{16 G_\infty^2}{\lambda_1 T} .
\end{equation*}
\end{theorem}

\subsubsection{Upper bounds for stochastic zeroth-order (derivative-free) optimization}\label{sec:Derivative-free results}
Compared to stochastic first-order optimization, less is known about rigorous upper bounds for stochastic zeroth-order optimization. The few rigorous upper bounds for zeroth-order optimization that are known \cite{flaxman2005online,agarwal2010optimal,agarwal2011stochastic,jamieson2012query,shamir2013complexity} are usually weaker than their first-order counterparts. The upper bounds we use in this paper are from  \cite{flaxman2005online} and \cite{agarwal2011stochastic}, in which the authors prove rigorous upper bounds for stochastic zeroth-order convex optimization in which the expected error in objective function value converges to zero like $\sqrt[4]{\frac{p^2}{T}}$ and $\sqrt{\frac{p^{32}}{T}}$, respectively, where $T$ is the number of iterations. We record their results below, adapted for our purposes. Note that in these papers, the authors state the results in terms of online optimization. However, it is straightforward to convert these into results for stochastic optimization. Similar adaptations of these same results are also noted in \cite{jamieson2012query} and \cite{shamir2013complexity}.

\begin{theorem}[Adapted from Theorem 2 of \cite{flaxman2005online} and Theorem 2 of \cite{agarwal2011stochastic}]\label{thm:DF}
Let the feasible set $\mathcal{X}\subset \mathbb{R}^p$ be contained in a Euclidean ball of radius $R_2$, and contain a Euclidean ball of radius $r_2$. Assume access to a stochastic oracle which, upon input $\vb{x}\in \mathcal{X}$, outputs a real-valued random variable $\hat{k}(\vb{x})$ such that $\E \hat{k}(\vb{x}) = f(\vb{x})$ and $| \hat{k}(\vb{x}) | \leq E$. If $f$ is convex and $L$-Lipschitz w.r.t. $\|\cdot \|_2$, there exists an algorithm \cite{flaxman2005online} that makes $T$ queries and outputs a (random) vector $\bar{\vb{x}}$ such that
\begin{equation*}
\E f(\bar{\vb{x}}) - f(\vb{x}^*) \leq O\qty(\frac{p^2 E^2 R_2^2( L + E/r_2)^2}{T} )^{1/4}.
\end{equation*}
There also exists an algorithm \cite{agarwal2011stochastic} that makes $T$ queries and outputs $\bar{\vb{x}}$ such that
\begin{equation*}
\E f(\bar{\vb{x}}) - f(\vb{x}^*) \leq \widetilde{O}\qty(\frac{p^{32} E^2 }{ T})^{1/2}.
\end{equation*}
\end{theorem}

Note that this implies that $\min\qty(O\qty(\frac{p^2 E^2 R_2^2( L + E/r_2)^2}{\epsilon^4}), \widetilde{O}(\frac{p^{32}E^2}{\epsilon^2}))$ queries are needed to optimize to expected precision $\epsilon$. Other rigorous upper bounds for derivative-free stochastic convex optimization exist \cite{agarwal2010optimal,jamieson2012query,shamir2013complexity}, which apply in settings which are less applicable to this paper, or to specific families of functions.

\section{Black-box formulation}\label{sec:black-box}

One of our main goals is to prove a rigorous lower bound on the number of quantum measurements required to minimize an objective function $f(\bm{\uptheta}) = \expval{H}{\bm{\uptheta}}$ to within some precision $\epsilon$. However, examining the formulation of this question already reveals a subtlety. Namely, a classical computer can simulate the quantum part of a variational algorithm in (in general) exponential time, which implies that any variational algorithm can be simulated with a variational algorithm which makes zero quantum measurements. 

Of course, in practice it would generally be intractable for a classical computer to simulate a variational algorithm. In fact, it is known that the existence of an efficient classical algorithm for sampling from the output distribution of the commonly-considered variational algorithm QAOA would imply a collapse of the polynomial hierarchy \cite{farhi2016quantum}. We therefore seek a formulation of the problem which better captures the behavior of realistic algorithms. One way to do this is to strip away the classical outer loop's knowledge of the specific objective observable $H$ that it is trying to optimize by encoding the observable in the black box, and only giving the outer loop the black box $\mathcal{O}_H$ and a promise that $H$ belongs to some particular family $\mathcal{H}$. In practice, this essentially means that the black-box picture is applicable when the classical optimization algorithm that the outer loop runs does not depend on the details of the objective observable $H$ itself (but could depend on the family $\mathcal{H}$), but is rather some more general-purpose algorithm like gradient descent, Nelder-Mead, SPSA, BFGS, etc. We are not aware of any proposed variational algorithm that is not in this class. Note that, while the classical optimization component of a variational algorithm does not exploit detailed structure of the objective observable in current proposals, the variational ansatz sometimes does (e.g. \cite{peruzzo2014variational,farhi2014quantum,wecker2015progress}). For example, in QAOA, the ansatz involves pulses of the form $e^{i H \bm{\uptheta}_j}$ where $H$ is the objective observable. For such  cases, our query upper bounds are still fully applicable. Our lower bound for zeroth-order variational optimization is ansatz-independent in the sense that for any ansatz the algorithm chooses, the lower bound still holds. Hence, the zeroth-order lower bound is still fully applicable in the setting where the ansatz may be a function of $H$. However, there is a subtlety that if the algorithm is promised that the ansatz is a certain function of $H$, it could use this information to learn the ground state of $H$ with fewer queries, and the lower bound may no longer apply. Essentially, this means that the zeroth-order lower bound applies for zeroth-order algorithms which do not cleverly exploit the dependence of the variational ansatz on the objective observable. This includes all ``general-purpose'' zeroth-order methods such as SPSA, Nelder-Mead, and gradient descent with gradients estimated via finite-differences, to name a few. 

In the remainder of this section, we define our black-box formulation of variational algorithms. Specifically, we define an oracle $\mathcal{O}_H$ encoding the objective observable $H$. This definition allows us to talk about the query complexity of optimizing some family of objective observables $\mathcal{H}$. The classical outer loop is given as input an oracle $\mathcal{O}_H$ under the promise that $H \in \mathcal{H}$, and attempts to find an approximate optimum using as few queries as possible. 

We sometimes refer to $\mathcal{O}_H$ as a ``sampling oracle'' for $H$, because it may be viewed as internally expanding $H$ and the pulse generators $A_i$ as linear combinations of tensor products of Paulis, and then randomly sampling some terms from the expansions to be measured, via some natural importance sampling algorithm. The oracle can be asked to output a random variable which corresponds to an estimate of the objective function, which we refer to as zeroth-order sampling, or output an estimate of a $k\textsuperscript{th}$ order derivative of the objective function, which we refer to as $k\textsuperscript{th}$ order sampling. These oracles can be straightforwardly implemented in practice.  

\begin{definition}[Sampling oracle]
Let $H$ be an objective observable. Then the sampling oracle $\mathcal{O}_H$ encoding $H$ is defined as follows. It receives as input a description of a $p$-parameter parameterization $\Theta$, a parameter $\bm{\uptheta}\in \mathbb{R}^p$, and a multiset $S$ that we call a ``coordinate multiset'' containing integers from the set $[p]$. If $|S| = k$, then the oracle internally follows the procedure for $k\textsuperscript{th}$ order sampling defined below, which returns some random variable $X$ such that $\E X = f(\bm{\uptheta}) := \expval{H}{\bm{\uptheta}}$ for zeroth-order sampling, or $\E X = \frac{\partial^k f}{\partial \bm{\uptheta}_{s_1} \cdots \partial \bm{\uptheta}_{s_k}} (\bm{\uptheta})$ for $k\textsuperscript{th}$ order sampling if $S = \{ s_1, \dots, s_k \}$. 
\end{definition}

If the parameterization $\Theta$ is clear from context, we may not explicitly note that $\Theta$ is provided as input to the oracle. If we speak of querying the oracle with a state $\ket{\psi}$, we mean querying the oracle with a parameterization and parameter vector that describe the state $\ket{\psi}$. In the remainder of this section, we first define how the oracle behaves for zeroth-order queries. We then briefly review how to measure gradients (and higher-order derivatives) in low depth in variational algorithms, and then define the behavior of the oracle upon first- and higher-order queries.

\subsection{Zeroth-order sampling}\label{sec:Zeroth-order oracle}
Let $H$ be some objective observable. Decompose $H$ into a linear combination of $m$ products of Pauli operators as $H = \sum_{i=1}^m \alpha_i P_i$ where $\alpha_i > 0$. (The coefficients may all be assumed to be positive by absorbing the phase into the operator.) Now, defining the normalization factor $E := \sum_i \alpha_i$ and the probability distribution $p_i := \alpha_i / E$, we may write
\begin{equation*}
H = E \sum_{i=1}^m p_i P_i = E \E_{i\sim p_i} P_i.
\end{equation*}

By linearity,
\begin{equation*}
f(\bm{\uptheta}) := \expval{H}{\bm{\uptheta}} =  E \E_{i\sim p_i} \expval{P_i}{\bm{\uptheta}}.
\end{equation*}

From this expression, it is clear that by sampling index $i$ with probability $p_i$, measuring $P_i$ with respect to the state $\ket{\bm{\uptheta}}$, and then multiplying the outcome by $E$ we obtain an unbiased estimator for $f(\bm{\uptheta}) = \expval{H}{\bm{\uptheta}}$. Furthermore, since the measurement outcome of $P_i$ is either $+1$ or $-1$, the output of this estimator is $\pm E$-valued. It is also clear that $| f(\bm{\uptheta}) | \leq E$ for all $\bm{\uptheta}$. We define the behavior of the sampling oracle $\mathcal{O}_H$ for zeroth-order sampling to be essentially the above process. 

\begin{definition}[Zeroth-order behavior of $\mathcal{O}_H$]\label{def:zeroth}
Let $H = E \sum_{i=1}^m p_i P_i$ be a decomposition of an objective observable as above, where $E > 0$ and $p_i$ is a probability distribution. Given as input a parameterization $\Theta$, a parameter $\bm{\uptheta}$, and an empty coordinate multiset $S = \varnothing$, the oracle behaves as follows. It internally prepares $\ket{\bm{\uptheta}}$ and measures the observable $P_i$ with probability proportional to $p_i$. It then multiplies the outcome by $E$ and outputs the resulting $\pm E$-valued estimator.
\end{definition}

\subsection{Analytic gradient measurements}
Variational algorithms typically aim to estimate the objective function $f(\bm{\uptheta}) := \expval{H}{\bm{\uptheta}}$ at some point $\bm{\uptheta}$ in parameter space, and use this information along with previous estimates of $f$ to propose a new point $\bm{\uptheta}^\prime$ in parameter space. In this case, the classical outer loop essentially has a stochastic zeroth-order oracle for the objective function. 

Some recent works \cite{guerreschi2017practical,mitarai2018quantum,romero2018strategies,schuld2018circuit,schuld2018evaluating} have instead suggested a different optimization strategy, in which one directly extracts information about the gradient of the objective function $f(\bm{\uptheta})$ by measuring corresponding quantum observables. Quantum measurements of this type that correspond to estimates of the gradient of the objective function are often referred to as \emph{analytic gradient measurements}. In this section, we review these strategies.

Consider a particular variational ansatz 
\begin{equation*}
\ket{\bm{\uptheta}} = \ket{\bm{\uptheta}_1, \dots, \bm{\uptheta}_p} = e^{-iA_p \bm{\uptheta}_p/2} \cdots e^{-iA_1 \bm{\uptheta}_1 / 2}\ket{\Psi}.
\end{equation*}
 For notational convenience, we define $U_i := e^{-i A_i \bm{\uptheta}_i / 2}$ to be the unitary corresponding to pulse $i$, and for $i \leq j$ we define $U_{i:j} := e^{-i A_j \bm{\uptheta}_j / 2} \cdots e^{-i A_i \bm{\uptheta}_i / 2}$ to be the sequence of pulses from $i$ through $j$, inclusive. Note that in using this notation we are hiding the dependence on $\bm{\uptheta}$ for visual clarity. Recall that the objective function corresponding to objective observable $H$ is given by
\begin{equation*}
f(\bm{\uptheta}) := \expval{H}{\bm{\uptheta}} := \expval{e^{i A_1 \bm{\uptheta}_1 / 2}\cdots e^{i A_p \bm{\uptheta}_p/2} H e^{-i A_p \bm{\uptheta}_p/2}\cdots e^{-i A_1 \bm{\uptheta}_1 / 2}}{\Psi} :=  \expval{U_{1:p}^\dagger H U_{1:p}}{\Psi}.
\end{equation*}

It is straightforward to calculate the following relation via the chain rule applied to the above expression:

\begin{equation*}
\pdv{f}{\bm{\uptheta}_j} (\bm{\uptheta}) = - \Im \expval{U_{1:j}^\dagger A_j U_{(j+1):p}^\dagger H U_{1:p}}{\Psi}.
\end{equation*}

We now describe how the above quantity could be measured in a variational algorithm. Denote the Pauli decomposition of $A_j$ as $A_j = \sum_{k=1}^{n_j} \beta^{(j)}_k Q^{(j)}_k$ where $Q^{(j)}_k$ are products of Pauli operators. As in the previous sections, denote the Pauli decomposition of $H$ as $H = \sum_{i=1}^m \alpha_i P_i$. Then by linearity we can rewrite the above derivative as
\begin{equation*}\label{eq:gradExpansion}
\pdv{f}{\bm{\uptheta}_j} (\bm{\uptheta}) = - \sum_{k=1}^{n_j}\sum_{l=1}^m \beta^{(j)}_k \alpha_l \Im \expval{U_{1:j}^\dagger Q^{(j)}_{k} U_{(j+1):p}^\dagger P_l U_{1:p}}{\Psi}.
\end{equation*}

Now, we can obtain an unbiased estimator for $\Im \expval{U_{1:j}^\dagger Q^{(j)}_{k} U_{(j+1):p}^\dagger P_l U_{1:p}}{\Psi}$ via a (generalized) Hadamard test. In particular, the following procedure may be used for estimating \\ $\Im \expval{U_{1:j}^\dagger Q^{(j)}_{k} U_{(j+1):p}^\dagger P_l U_{1:p}}{\Psi}$.

\vbox{%
\noindent\fbox{%
    \parbox{\textwidth}{
	\underline{\textbf{Hadamard test for estimating $-\Im \expval{U_{1:j}^\dagger Q^{(j)}_{k} U_{(j+1):p}^\dagger P_l U_{1:p}}{\Psi}$}}
	
	\begin{enumerate}
		\item Initialize Register $A$ in the qubit state $\ket{+}_A$. Initialize Register $B$ in the state $\ket{\Psi}_B$.
		\item Apply $U_{1:j}$ to Register $B$.
		\item Apply a Controlled-$Q^{(j)}_k$ gate to Register $B$, controlled on Register $A$.
		\item Apply $U_{(j+1):p}$ to Register $B$.
		\item Apply a Controlled-$P_l$ gate to Register $B$, controlled on Register $A$.
		\item Measure the Pauli $Y$ operator on Register $A$.
	\end{enumerate}
	
	The above procedure yields a $\pm 1$-valued unbiased estimator for $-\Im \expval{U_{1:j}^\dagger Q^{(j)}_{k} U_{(j+1):p}^\dagger P_l U_{1:p}}{\Psi}$, requiring one quantum measurement.
	
    }%
}
\begin{center}
Algorithm \hypertarget{alg1}{1}: Generalized Hadamard test \cite{ekert2002direct,li2017efficient,guerreschi2017practical,romero2018strategies}
\end{center}}

Hence, one may estimate $\nabla f (\bm{\uptheta})$ by expanding the derivatives as above, and then estimating each term of the expansion using Algorithm \hyperlink{alg1}{1}. Alternative methods of analyticly measuring derivatives are described in \cite{mitarai2018quantum,schuld2018evaluating}, which require similar quantum resources to the scheme we just described (but do not necessarily require controlled-Pauli gates). We note that throughout this paper, one could estimate gradients using a strategy based on these methods instead, and the results would be essentially unchanged.

We now describe an equivalent way of understanding analytic gradients. Observe that
\begin{equation*}
\pdv{f}{\bm{\uptheta}_j} (\bm{\uptheta}) = - \Im \expval{U_{1:j}^\dagger A_j U_{(j+1):p}^\dagger H U_{1:p}}{\Psi} = \frac{1}{2}\expval{i [U_{(j+1):p}A_j U_{(j+1):p}^\dagger , H]}{\bm{\uptheta}}.
\end{equation*}
Hence, if we define the Hermitian operators
\begin{equation*}
G_j := \frac{i}{2} [U_{(j+1):p}A_j U_{(j+1):p}^\dagger , H],
\end{equation*}
and we define $\vec{G} := (G_1, \dots, G_p)^{\top}$, then we may write $\nabla f(\bm{\uptheta}) = \expval{\vec{G}}{\bm{\uptheta}}$. An alternative commutator expression for the derivatives was noted in \cite{mcclean2018barren}.

\subsection*{The case of a constraint $\bm{\uptheta}_i = \bm{\uptheta}_j$}
There are cases in which one may want to impose a constraint that some parameters are always equal. For example, this situation occurs for the ``Hamiltonian variational'' ansatz proposed in \cite{wecker2015progress}. Hence, one could have a $p$-pulse ansatz but a smaller number of independent variational parameters. We note that this situation is easily addressed within the framework of this paper. For example, consider the case in which $\bm{\uptheta}_i$ is constrained to always equal $\bm{\uptheta}_j$, i.e. $\bm{\uptheta}_i = \bm{\uptheta}_j := \xi$. It is straightforward to show by linearity that $\pdv{f}{\xi} (\bm{\uptheta}) = \expval{(G_i + G_j)}{\bm{\uptheta}}$, where $G_i$ and $G_j$ are defined as above. Hence, $\pdv{f}{\xi} (\bm{\uptheta})$ may be estimated via Algorithm \hyperlink{alg1}{1} just as in the unconstrained case. For simplicity, we assume that there are no such constraints on the parameters. However, all results in this paper can be easily generalized to work with such constraints via this observation. 

\subsection{First-order sampling}\label{sec:first-order sampling}
In the previous section, we described how information about the derivatives of the objective function can be extracted in low depth using a generalized Hadamard test. In this section, we describe a specific estimator of a derivative of the objective function which requires one Pauli measurement. We will use this estimator to define the behavior of the oracle $\mathcal{O}_H$ upon a first-order query. 

As in the previous section, denote the Pauli expansion of $H$ as $H = \sum_{i=1}^m \alpha_i P_i$ and the Pauli expansion of $A_j$ as $A_j = \sum_{k=1}^{n_j} \beta^{(j)}_k Q^{(j)}_k$, where all $\alpha$ and $\beta$ coefficients are positive real numbers. Then we may write $\pdv{f}{\bm{\uptheta}_j} (\bm{\uptheta})$ as the following expansion:

\begin{equation*}\label{eqn:gradExpansion}
\pdv{f}{\bm{\uptheta}_j} (\bm{\uptheta}) = \sum_{k=1}^{n_j} \sum_{l=1}^{m} \beta_{k}^{(j)} \alpha_l \expval{\frac{i}{2} \comm{U_{(j+1):p} Q^{(j)}_{k} U_{(j+1):p}^\dagger}{P_l} }{\bm{\uptheta}}.
\end{equation*}

We now rewrite this expansion as a certain expectation value, similarly to what we did in the definition of zeroth-order sampling. First, we observe that some of the commutators in the expansion  may trivially be zero, if the operators $U_{(j+1):p} Q^{(j)}_{k} U_{(j+1):p}^\dagger$ and $P_l$ act nontrivially on disjoint sets of qubits. This will often be the case in the toy model we analyze in Section \ref{sec:Separation}. Removing terms that are trivially zero will improve convergence in our optimization algorithms. To this end, we define a new set of coefficients:

\begin{equation*}\label{eqn:gamma}
\gamma^{(j)}_{kl} :=
    \begin{cases}
        0, & \text{qubits}\qty( U_{(j+1):p} Q^{(j)}_{k} U_{(j+1):p}^\dagger )\cap  \text{qubits}(P_l) = \varnothing \\
        \beta^{(j)}_k \alpha_l, & \text{qubits}\qty(U_{(j+1):p} Q^{(j)}_{k} U_{(j+1):p}^\dagger) \cap \text{qubits}\qty(P_l ) \neq \varnothing
    \end{cases}
\end{equation*}

\noindent where $\text{qubits}(U_{(j+1):p} Q^{(j)}_{k} U_{(j+1):p}^\dagger)$ denotes the set of qubits on which $U_{(j+1):p} Q^{(j)}_{k} U_{(j+1):p}^\dagger$ acts nontrivially, after removing pulses which trivially commute through $Q^{(j)}_k$ and cancel the corresponding inverse pulse.  We define the associated normalization factors 
$$\Gamma_j = \sum_{k=1}^{n_j} \sum_{l=1}^m \gamma^{(j)}_{kl},$$
 and probability distributions $q^{(j)}_{kl} := \frac{1}{\Gamma_j} \gamma^{(j)}_{kl}$ over the indices $k$ and $l$, where $j$ is considered fixed. Note that we have the bound $\Gamma_j \leq EB_j$ where $B_j := \sum_{k=1}^{n_j}\beta^{(j)}_k$. Equipped with these definitions, we may write
\begin{equation*}
\pdv{f}{\bm{\uptheta}_j} (\bm{\uptheta}) = \Gamma_j \E_{(K,L) \sim q^{(j)}_{KL} } \expval{\frac{i}{2} \comm{U_{(j+1):p} Q^{(j)}_{K} U_{(j+1):p}^\dagger}{P_L} }{\bm{\uptheta}}.
\end{equation*}

It is straightforward to see that  $| \pdv{f}{\bm{\uptheta}_j} (\bm{\uptheta}) | \leq \Gamma_j$ for all $\bm{\uptheta}$. Given the above representation of $\pdv{f}{\bm{\uptheta}_j} (\bm{\uptheta})$, it is clear that the following procedure provides an unbiased estimator for $\pdv{f}{\bm{\uptheta}_j} (\bm{\uptheta})$ which requires a single measurement.

\noindent\fbox{%
    \parbox{\textwidth}{%
        \underline{\textbf{An unbiased one-measurement estimator for $\pdv{f}{\bm{\uptheta}_j} (\bm{\uptheta})$.}}
        
        \begin{enumerate}
            \item Sample $(K,L)$ from the distribution $q^{(j)}_{KL}$ as defined above.
            \item Use a Hadamard test (Algorithm \hyperlink{alg1}{1}) to obtain a one-measurement unbiased estimate of $\expval{\frac{i}{2} \comm{U_{(j+1):p} Q^{(j)}_{K} U_{(j+1):p}^\dagger}{P_L} }{\bm{\uptheta}} = -\Im \expval{U_{1:j}^\dagger Q^{(j)}_K U^\dagger_{(j+1):p} P_L U_{1:p}}{\Psi}$.
            \item Multiply the resulting number by $\Gamma_j$.
        \end{enumerate}
        
        The estimator for $\pdv{f}{\bm{\uptheta}_j} (\bm{\uptheta})$ described above is $\pm \Gamma_j$-valued.
    }%
}
\begin{center}
\vspace{-1ex}Algorithm \hypertarget{alg2}{2}: unbiased, one-measurement estimator for $\pdv{f}{\bm{\uptheta}_j} (\bm{\uptheta})$.
\end{center}

Motivated by these derivative-estimating procedures, we now define the first-order behavior of the oracle $\mathcal{O}_H$.

\begin{definition}[First-order behavior of $\mathcal{O}_H$]
Let $H$ denote an objective observable. Upon input of parameterization $\Theta$, parameter $\bm{\uptheta}$, and a coordinate multiset $S = \{ j \}$ for some $j\in [p]$, the oracle internally prepares the state $\ket{\bm{\uptheta}}$ and runs Algorithm \hyperlink{alg2}{2} above. It outputs the resulting $\pm \Gamma_j$-valued  estimator for $\pdv{f}{\bm{\uptheta}_j} (\bm{\uptheta})$. 
\end{definition}

\subsection{Higher-order sampling}\label{sec:higher}
The sampling procedure we have described above for obtaining unbiased estimates of derivatives in low depth generalizes to higher-order derivatives. In this section, we outline how the procedure  would work. Start by recalling the derivative operators we derived above:
\begin{equation*}
G_j := \frac{i}{2} [U_{(j+1):p}A_j U_{(j+1):p}^\dagger , H].
\end{equation*}
To compress notation, we define the Hermitian operators $\tilde{A}_j := U_{(j+1):p}A_j U_{(j+1):p}^\dagger $ so that $G_j = \frac{i}{2}\comm{\tilde{A}_j}{H}$ and $\pdv{f}{\bm{\uptheta}_j} (\bm{\uptheta}) = \expval{G_j}{\bm{\uptheta}}$. Note that the operator $G_j$ is independent of $\bm{\uptheta}_k$ for $k \leq j$. Hence, if we take the partial derivatives of both sides of the above expression with respect to $\bm{\uptheta}_k$ with $k \leq j$, we get
\begin{equation*}
\pdv[2]{f}{\bm{\uptheta}_k}{\bm{\uptheta}_j} (\bm{\uptheta}) =  \expval{\frac{i}{2}\comm{\tilde{A}_k}{G_j}}{\bm{\uptheta}},
\end{equation*}
where, since $G_j$ is independent of $\bm{\uptheta}_k$, this result follows from arguments identical to those we used to derive the expression for $G_j$. Also, note that from the original definition $\expval{H}{\bm{\uptheta}} := \expval{e^{i A_1 \bm{\uptheta}_1 / 2}\cdots e^{i A_p \bm{\uptheta}_p/2} H e^{-i A_p \bm{\uptheta}_p/2}\cdots e^{-i A_1 \bm{\uptheta}_1 / 2}}{\Psi}$, it is clear that $\pdv[2]{f}{\bm{\uptheta}_k}{\bm{\uptheta}_j} = \pdv[2]{f}{\bm{\uptheta}_j}{\bm{\uptheta}_k} $. We therefore have, for $k \leq j$,
\begin{equation*}
\pdv[2]{f}{\bm{\uptheta}_k}{\bm{\uptheta}_j} = \pdv[2]{f}{\bm{\uptheta}_j}{\bm{\uptheta}_k} = -\frac{1}{4} \expval{\comm{\tilde{A}_k}{\comm{\tilde{A}_j}{H}}}{\bm{\uptheta}}.
\end{equation*}

\noindent To see how to estimate this in low depth, note that we have

\begin{equation*}
\pdv[2]{f}{\bm{\uptheta}_k}{\bm{\uptheta}_j} = \frac{1}{2}\Re \qty( -\expval{\tilde{A}_k \tilde{A}_j H}{\bm{\uptheta}} +\expval{\tilde{A}_k H \tilde{A}_j}{\bm{\uptheta}}  ).
\end{equation*}

From the above expression, we see how to generalize the first-order sampling procedure to higher orders. To obtain an unbiased estimate of $\pdv[2]{f}{\bm{\uptheta}_k}{\bm{\uptheta}_j}$ with a single  measurement, first expand $A_k$, $A_j$, and $H$ as linear combinations of products of Paulis. In turn, this yields an expansion of $\pdv[2]{f}{\bm{\uptheta}_k}{\bm{\uptheta}_j}$ as a linear combination of real parts of inner products of states that are acted on with pulses and Paulis. For a one-measurement estimator, randomly choose one of these inner products with probability proportional to the magnitude its coefficient, and then get an unbiased estimate of the inner product by performing a Hadamard test, similarly to what we described for  the first-order case. Note that, in the second-order case (or more generally for the even-order case), the Hadamard test will involve an $X$-basis measurement instead of $Y$-basis measurement, since a real part is being estimated. 

This procedure works in general for $k\textsuperscript{th}$ order derivatives. In particular, the observable corresponding to a $k\textsuperscript{th}$ order derivative will be a nested commutator of depth $k$.

\subsection{Query complexity in the black-box formalism}

Having defined the sampling oracle $\mathcal{O}_H$, we may now quantify the cost of an algorithm by the number of queries it makes to the oracle. The general setup for a variational optimization problem in the black-box setting is that the classical ``outer loop'' is promised that  the objective observable $H$ to be minimized belongs to a family $\mathcal{H}$ of observables, and is given access to the sampling oracle $\mathcal{O}_H$. Note that, from the perspective of the outer loop, the problem of minimizing the objective function is a purely classical black-box optimization problem since it gives classical input to the oracle and receives classical output. We now formalize the notion of the ``error'' associated with some variational algorithm $\mathcal{A}$ for optimizing a family $\mathcal{H}$ of objective observables.

\begin{definition}
Let $\mathcal{H}$ denote a set of objective observables, and $\mathcal{A}$ be a (possibly randomized) classical algorithm which has access to a sampling oracle $\mathcal{O}_H$ for some $H\in \mathcal{H}$ and outputs a description of a quantum state $\ket{\psi}$. Then the optimization error of $\mathcal{A}$ with respect to $\mathcal{H}$, $\Err(\mathcal{A},\mathcal{H})$, is defined to be
\begin{equation*}
\Err(\mathcal{A},\mathcal{H}) := \sup_{\mathcal{O}_H\, :\, H \in \mathcal{H}} \E_{\psi} \qty[ \expval{H}{\psi} - \lambda_{\text{min}}(H) ]
\end{equation*}
where the expectation is over the possible randomness of the output state $\ket{\psi}$.
\end{definition}
In other words, $\Err(\mathcal{A},\mathcal{H})$ is the worst-case expected error in objective function value that $\mathcal{A}$ makes over all objective observables in the set $\mathcal{H}$. We now make a few more definitions that will be convenient later.

\begin{definition}[$k\textsuperscript{th}$-order algorithm]
We say that a black-box algorithm $\mathcal{A}$ is a $k\textsuperscript{th}$-order algorithm if it makes at most $k\textsuperscript{th}$-order queries to the oracle. That is, if the coordinate multiset $S$ provided to the oracle satisfies $|S| \leq k$ for each query.
\end{definition}

\begin{definition}[$\delta$-vicinity algorithm]
Define the $\delta$-optimum of an observable $H$ to be the set of all states $\ket{\psi}$ such that $\expval{H}{\psi} - \lambda_{\min}(H) \leq \delta$. Define the $\delta$-optimum of a set of observables $\mathcal{H}$ to be the union of the $\delta$-optima of each observable in the set. We say a black-box algorithm $\mathcal{A}$ is a $\delta$-vicinity algorithm for $\mathcal{H}$ if it only queries the black box with descriptions of states that are in the $\delta$-optimum of $\mathcal{H}$.
\end{definition}

\section{General upper bounds for variational algorithms in a convex region}\label{sec:General}
In this section, we give general upper bounds on the query cost of variational algorithms in a region where the objective function is convex. This amounts to applying the known upper bounds for stochastic convex optimization from Section \ref{sec:Requisite upper bounds} to the setting in which estimates of the objective function, or derivatives of the objective function, come from the oracle specified in Section \ref{sec:black-box} (which is easy to implement in low depth in practice). Note that the oracle returns estimates of partial derivatives w.r.t. specific components. However, there are multiple ways of using these derivative estimates to construct a gradient estimator. We describe two such estimators. The first is designed to be used with SGD, and the second is designed to be used with SMD with an $l_1$ setup.

For the remainder of this section, fix some objective observable with Pauli expansion $H = \sum_{i=1}^m \alpha_i P_i$ and some parameterization $\Theta$ whose pulse generators $A_j$ have Pauli expansions $A_j = \sum_{k=1}^{n_j} \beta^{(j)}_k Q^{(j)}_k$. As in Section \ref{sec:black-box}, define $E := \sum_{i=1}^m \alpha_i$, $B_j = \sum_{k=1}^{n_j} \beta^{(j)}_k$. Define $\Gamma_j$ to be the  normalization factor associated with coordinate $j$ as defined in Section \ref{sec:black-box} (see also Table \ref{tab:notation}). We collect these $\Gamma_j$ into a vector as
$$\vec{\Gamma} := (\Gamma_1, \dots, \Gamma_p)^\top.$$
\subsection{Gradient estimators from oracle queries}

First, we specify some unbiased estimators for the gradient that we will use. The estimator of Algorithm \hyperlink{alg3}{3} is based on $l_1$ sampling and designed with the goal in mind of achieving a smaller $2$-norm of the estimator and will be used in conjunction with SGD. The estimator of Algorithm \hyperlink{alg4}{4} is based on $l_2$ sampling and designed with the goal of achieving a smaller $\infty$-norm of the estimator and will be used in conjunction with SMD. The latter estimator also requires a mild assumption on the $\infty$-norm of the objective function.   The estimators also differ in their number of samples: the former estimator uses a single sample while the latter could be called a ``mini-batch'' estimator which uses an asymptotically growing number of samples.
We first define the two estimators, and then prove their correctness and bound them in the subsequent lemmas.

\noindent\fbox{%
    \parbox{\textwidth}{%
        \underline{\textbf{An unbiased one-query estimator for $\nabla f(\bm{\uptheta})$.}}
        
        \begin{enumerate}
            \item Select coordinate $j$ with probability $\frac{\Gamma_j}{\|\vec{\Gamma}\|_1}$. 
            \item Query $\mathcal{O}_H$ with parameter $\bm{\uptheta}$ and coordinate multiset $\{j\}$.
            \item Multiply the output of the oracle by $\frac{\|\vec{\Gamma}\|_1}{\Gamma_j} \hat{e}_j$.
        \end{enumerate}
        
        The above procedure outputs a vector $\hat{\vb{g}}(\bm{\uptheta})$ such that $\E \hat{\vb{g}}(\bm{\uptheta}) = \nabla f(\bm{\uptheta})$ and $\| \hat{\vb{g}}(\bm{\uptheta}) \| = \|\vec{\Gamma}\|_1$.
    }%
}
\begin{center}
\vspace{-1ex}Algorithm \hypertarget{alg3}{3}: $l_1$-sampling estimator for $\nabla f(\bm{\uptheta})$.
\end{center}

\noindent\fbox{%
    \parbox{\textwidth}{%
        \underline{\textbf{An unbiased $\widetilde{O}(p)$-query estimator for $\nabla f(\bm{\uptheta})$.}}
        
        For each $j \in [p]$, query $\mathcal{O}_H$ with parameter $\bm{\uptheta}$ and coordinate multiset $\{j\}$ $N_j$ times, where $N_j = \left\lceil p \frac{\Gamma_j^2}{\|\vec{\Gamma}\|_2^2} \ln\qty(4p^2\frac{\|\vec{\Gamma}\|_\infty^2}{\|\vec{\Gamma}\|_2^2})\right\rceil$. Letting $\hat{G}_j$ denote the average of the $N_j$ oracle outputs corresponding to component $j$, output $\hat{\vb{g}}(\bm{\uptheta}) = \sum_{i=1}^p \hat{G}_i \hat{e}_i$. \\
        
        Assuming $\|\nabla f(\bm{\uptheta})\|_\infty \leq \frac{\|\vec{\Gamma}\|_2}{\sqrt{2p}}$, the above procedure outputs a vector $\hat{\vb{g}}(\bm{\uptheta})$ such that $\mathbb{E}\hat{\vb{g}}(\bm{\uptheta}) = \nabla f(\bm{\uptheta})$ and $\mathbb{E} \| \hat{\vb{g}}(\bm{\uptheta}) \|_\infty^2 \leq \frac{5\|\vec{\Gamma}\|_2^2}{2p}$, while requiring at most $N = \sum_{j=1}^p N_j \leq p\qty[ 1 + \ln\qty(4p^2 \frac{\|\vec{\Gamma}\|_\infty^2}{\| \vec{\Gamma}\|_2^2})]$ samples.
    }%
}
\begin{center}
\vspace{-1ex}Algorithm \hypertarget{alg4}{4}: $l_2$-sampling estimator for $\nabla f(\bm{\uptheta})$.
\end{center}

\begin{lemma}[Correctness of Algorithm \hyperlink{alg3}{3}]
Algorithm \hyperlink{alg3}{3} outputs a vector $\hat{\vb{g}}(\bm{\uptheta})$ such that $\E \hat{\vb{g}}(\bm{\uptheta}) = \nabla f(\bm{\uptheta})$ and $\| \hat{\vb{g}}(\bm{\uptheta}) \| = \|\vec{\Gamma}\|_1$.
\end{lemma}
\begin{proof}
Recalling that the output of $\mathcal{O}_H$ upon querying a first-order derivative of the $j$th component is $\pm \Gamma_j$, it is clear that the vector output by the above procedure will have norm $\|\vec{\Gamma}\|_1$. Now we show that $\E \hat{\vb{g}}(\bm{\uptheta}) = \nabla f (\bm{\uptheta})$. Recall that the probability of selecting index $j$ is $\Gamma_j / \|\vec{\Gamma}\|_1$, and conditioned on index $j$ being selected, the expected output of the procedure is $\frac{\|\vec{\Gamma}\|_1}{\Gamma_j} \pdv{f}{\bm{\uptheta}_j} (\bm{\uptheta}) \hat{e}_j$. It follows that $\E \hat{\vb{g}}(\bm{\uptheta}) = \nabla f(\bm{\uptheta})$ as desired.
\end{proof}

\begin{lemma}[Correctness of Algorithm \hyperlink{alg4}{4}]
Algorithm \hyperlink{alg4}{4} outputs a vector $\hat{\vb{g}}(\bm{\uptheta})$ such that $\E \hat{\vb{g}}(\bm{\uptheta}) = \nabla f(\bm{\uptheta})$ and $\mathbb{E} \| \hat{\vb{g}}(\bm{\uptheta}) \|_\infty^2 \leq \frac{5\|\vec{\Gamma}\|_2^2}{2p}$.
\end{lemma}
\begin{proof}
Since the output of $\mathcal{O}_H$ upon receiving as input the parameter $\bm{\uptheta}$ and derivative multiset $\{j\}$ is a $\pm \Gamma_j$-valued random variable with expectation $\pdv{f}{\bm{\uptheta}_j} (\bm{\uptheta})$, we have $\mathbb{E}\hat{\vb{g}}(\bm{\uptheta}) = \sum_{i=1}^p \mathbb{E} \hat{G}_i \hat{e}_i = \nabla f(\bm{\uptheta})$. 

We now seek to upper bound $\mathbb{E} \| \hat{\vb{g}} \|_\infty^2$. We first turn our attention to the distribution of the random variable $\hat{G}_j$. Since $\hat{G}_j$ is an average of i.i.d. $\pm \Gamma_j$-valued random variables, Hoeffding's inequality implies
\begin{equation*}
\Pr\qty(|\hat{G}_j - (\nabla f)_j | \geq t) \leq 2 \exp\qty(-\frac{2t^2 N_j}{\Gamma_j^2})
\end{equation*}
for $t \geq 0$. Using this bound with $t = \frac{\|\vec{\Gamma}\|_2}{\sqrt{2p}}$ and recalling $N_j = \left\lceil p \frac{\Gamma_j^2}{\|\vec{\Gamma}\|_2^2} \ln\qty(4p^2\frac{\|\vec{\Gamma}\|_\infty^2}{\|\vec{\Gamma}\|_2^2})\right\rceil$ yields
\begin{equation*}
\Pr\qty(|\hat{G}_j - (\nabla f)_j | \geq \frac{\|\vec{\Gamma}\|_2}{\sqrt{2p}}) \leq \frac{\|\vec{\Gamma}\|_2^2}{2p^2 \| \vec{\Gamma} \|_\infty^2}.
\end{equation*}

By the union bound, the probability that $|\hat{G}_j - (\nabla f)_j | \geq \frac{\|\vec{\Gamma}\|_2}{\sqrt{2p}}$ for some $j$ is upper bounded by $\frac{\|\vec{\Gamma}\|_2^2}{2p \| \vec{\Gamma} \|_\infty^2}$. If this event occurs, then we only have the trivial upper bound $\| \hat{\vb{g}} \|_\infty^2 \leq \| \vec{\Gamma} \|_\infty^2$. Conditioned on this ``bad'' event not occurring, we have the bound $\| \hat{\vb{g}} \|_\infty^2 \leq 2 \frac{\|\vec{\Gamma}\|_2^2}{p}$, where we used the assumption that $|(\nabla f)_j | \leq \frac{\|\vec{\Gamma}\|_2}{\sqrt{2p}}$ for all $j$. It follows that
\begin{equation*}
\mathbb{E} \| \hat{\vb{g}} \|_\infty^2 \leq 2 \frac{\|\vec{\Gamma}\|_2^2}{p} + \| \vec{\Gamma} \|_\infty^2\cdot \frac{\|\vec{\Gamma}\|_2^2}{2p \| \vec{\Gamma} \|_\infty^2} = \frac{5\|\vec{\Gamma}\|_2^2}{2p}.
\end{equation*}
\end{proof}

\subsection{Upper bounds}
Fix an objective observable $H$, parameterization $\Theta$, and a closed, convex feasible set $\mathcal{X}$. Define $f(\bm{\uptheta}) := \expval{H}{\bm{\uptheta}}$, and define $\vec{\Gamma}$ as above (see also Table \ref{tab:notation}). Let $\bm{\uptheta}^*$ denote a minimizer of $f(\bm{\uptheta})$ on $\mathcal{X}$.

\begin{lemma}[SGD bound]
If $f$ is convex on $\mathcal{X}$, and if $\mathcal{X}$ is contained in a Euclidean ball of radius $R_2$, then querying $\mathcal{O}_H$ with first-order queries and using the outputs to run projected SGD (with appropriate stepsizes) finds a (random) parameter $\bar{\bm{\uptheta}}$ such that $\E f(\bar{\bm{\uptheta}}) - f(\bm{\uptheta}^*) \leq \epsilon$ with $\frac{2 R_2^2 \|\vec{\Gamma}\|_1^2}{\epsilon^2}$ queries.
\end{lemma}
\begin{proof}
Use the 1-query estimator of Algorithm \hyperlink{alg3}{3} for $\nabla f(\bm{\uptheta})$ in conjunction with Theorem \ref{thm:SGD}.
\end{proof}

\begin{lemma}[SGD bound, strongly convex case]\label{lem:SGDSCv}
If $f$ is $\lambda_2$-strongly convex on $\mathcal{X}$ with respect to $\|\cdot \|_2$, then querying $\mathcal{O}_H$ with first-order queries and using the outputs to run projected SGD (with appropriate stepsizes) finds a parameter $\bar{\bm{\uptheta}}$ such that $\E f(\bar{\bm{\uptheta}}) - f(\bm{\uptheta}^*) \leq \epsilon$ with $\frac{2 \|\vec{\Gamma}\|_1^2}{\lambda_2 \epsilon}$ queries.
\end{lemma}
\begin{proof}
Use the 1-query estimator of Algorithm \hyperlink{alg3}{3} for $\nabla f(\bm{\uptheta})$ in conjunction with Theorem \ref{thm:SGDSC}.
\end{proof}

\begin{lemma}[SMD bound]
If $f$ is convex on $\mathcal{X}$, and if $\mathcal{X}$ is contained in a 1-ball of radius $R_1$, then querying $\mathcal{O}_H$ with first-order queries and using the outputs to run projected SMD with an appropriate $l_1$ setup and stepsizes finds a parameter $\bar{\bm{\uptheta}}$ such that $\E f(\bar{\bm{\uptheta}}) - f(\bm{\uptheta}^*) \leq \epsilon$ with $\frac{5e R_1^2 \|\vec{\Gamma}\|_2^2 \ln p}{\epsilon^2} \qty[1+\ln\qty(4p^2 \frac{\|\vec{\Gamma}\|_\infty^2}{\|\vec{\Gamma}\|_2^2})] = O\qty(\frac{R_1^2 \|\vec{\Gamma}\|_2^2 (\ln p)^2}{\epsilon^2})$ queries.
\end{lemma}
\begin{proof}
Use the estimator of Algorithm \hyperlink{alg4}{4} for $\nabla f(\bm{\uptheta})$, which outputs a gradient estimate $\hat{\vb{g}}$ such that $\mathbb{E}\| \hat{\vb{g}} \|_\infty^2 \leq \frac{5\|\vec{\Gamma}\|_2^2}{2p}$ and requires at most $p\qty[ 1 + \ln\qty(4p^2 \frac{\|\vec{\Gamma}\|_\infty^2}{\| \vec{\Gamma}\|_2^2})]$ queries per gradient estimate. Use these gradient estimates in conjunction with Theorem \ref{thm:SMD}.
\end{proof}

\begin{lemma}[SMD bound, strongly convex case]
If $f$ is $\lambda_1$-strongly convex on $\mathcal{X}$ with respect to $\| \cdot \|_1$, then querying $\mathcal{O}_H$ with first-order queries and using the outputs to run projected SMD with an appropriate choice of mirror map and stepsizes can find a parameter $\bar{\bm{\uptheta}}$ such that $\E f(\bar{\bm{\uptheta}}) - f(\bm{\uptheta}^*) \leq \epsilon$ with $\frac{40 \|\vec{\Gamma}\|_2^2}{\epsilon \lambda_1}  \qty[1+\ln\qty(4p^2 \frac{\|\vec{\Gamma}\|_\infty^2}{\|\vec{\Gamma}\|_2^2})] = O\qty(\frac{\|\vec{\Gamma}\|_2^2 \ln p}{\epsilon \lambda_1})$ queries.
\end{lemma}
\begin{proof}
Use the estimator of Algorithm \hyperlink{alg4}{4} for $\nabla f(\bm{\uptheta})$ in conjunction with Theorem \ref{thm:SMDSC}.
\end{proof}

For comparison, we also present an upper bound for the case in which we only make zeroth-order queries to $\mathcal{O}_H$.

\begin{lemma}[Derivative-free bound]\label{lem:DFv}
If $\mathcal{X}$ is a closed convex set contained in a Euclidean ball of radius $R_2$ and containing a Euclidean ball of radius $r_2$, and $f$ is $L_2$-Lipschitz w.r.t. $\|\cdot \|_2$ and convex on $\mathcal{X}$, then querying $\mathcal{O}_H$ with zeroth-order queries and using the outputs in conjunction with the algorithm of  \cite{flaxman2005online} or \cite{agarwal2011stochastic} finds a parameter $\bar{\bm{\uptheta}}$ such that $\E f(\bar{\bm{\uptheta}}) - f(\bm{\uptheta}^*) \leq \epsilon$ with $O\qty(\frac{p^2 E^2 R_2^2( L_2 + E/r_2)^2}{\epsilon^4})$ queries or $\widetilde{O}\qty(\frac{p^{32} E^2}{\epsilon^2})$ queries, respectively.
\end{lemma}
\begin{proof}
Note that the outputs of zeroth-order queries to $\mathcal{O}_H$ have magnitude $E$, and apply Theorem \ref{thm:DF}.
\end{proof}

We collect these results in Table \ref{tab:introUpper}.

\subsection{When is SMD superior to SGD?}
In Section \ref{sec:SummaryOfResults}, we gave intuition for why we might hope that using the 1-norm instead of 2-norm and using SMD with an $l_1$ setup instead of SGD might be beneficial in some cases. In particular, we noted that the 1-norm of a parameter vector $\bm{\uptheta}$ has a natural interpretation as the duration of evolution from the starting state $\ket{\Psi}$ to the trial state associated with $\bm{\uptheta}$, $\ket{\bm{\uptheta}}$. The $l_1$-distance between two parameter vectors may be interpreted as the amount of time for while the two associated pulse sequences differ. 

Comparing the upper bounds from the previous section, we see that where SGD has a factor of $\|\vec{\Gamma}\|_1^2$, SMD with an $l_1$ setup has instead a factor of $\|\vec{\Gamma}\|_2^2$.  Note that $\|\vec{\Gamma}\|_2^2$ is never larger than $\|\vec{\Gamma}\|_1^2$, and in fact can a factor of $p$ smaller. Consider for example the case in which $\Gamma_1 \approx \Gamma_2 \approx \cdots \approx \Gamma_p$, which may be a realistic scenario in practice. In this case, we have $\| \vec{\Gamma}\|_2^2 \approx p \Gamma_1^2$ for the SMD bound, whereas we have $\| \vec{\Gamma}\|_1^2 \approx p^2 \Gamma_1^2$ for the SGD bound, which is quadratically worse in the dimension of parameter space.

On the other hand, where the SGD bounds involve a factor of $R_2^2$, the SMD bounds involve a factor of $R_1^2$. $R_2^2$ is never larger than $R_1^2$, and can be significantly smaller. This could be the case when, for example, the feasible set $\mathcal{X}$ is a Euclidean ball. On the other hand, if $\mathcal{X}$ is a 1-ball, then $R_1 = R_2$, and SMD could potentially achieve substantially better performance than SGD due to the $\|\vec{\Gamma}\|_2^2$ versus $\|\vec{\Gamma}\|_1^2$ discrepancy. 

Another consideration is the issue of strong convexity. As is evident from the above bounds, the presence of strong convexity can substantially accelerate the optimization. SGD can take advantage of strong convexity w.r.t. the 2-norm, but SMD in the $l_1$ setup measures strong convexity w.r.t. the 1-norm, and in fact it is straightforward to show that the strong convexity parameters are related by $\lambda_1 \leq \lambda_2 \leq p \lambda_1$. In the toy problem we analyze in Section \ref{sec:Separation}, we have $\lambda_2 = \Theta(1)$ but $\lambda_1 = \Theta(1/n)$ where $n$ is the number of qubits. However, $\| \vec{\Gamma} \|_1^2 = \Theta(n^2)$ while $\| \vec{\Gamma} \|_2^2 = \Theta(n)$, so up to log factors and constants, SGD and SMD achieve the same asymptotic convergence rate for this toy model.

In conclusion, it is not clear from the upper bounds in the previous section or from the toy model we study in Section \ref{sec:Separation} whether SGD or SMD with an $l_1$ setup would typically achieve better upper bounds in practice. It is an interesting problem for future work to understand whether an $l_2$ (Euclidean) setup or an $l_1$ setup is usually more natural for variational algorithms.

\section{Oracle separation between zeroth-order and first-order optimization strategies for variational algorithms}\label{sec:Separation}
In this section, we prove a separation between algorithms which make only zeroth-order queries to the sampling oracle, and those which make first-order queries to the sampling oracle, within the vicinity of the global optimum. This separation is proven with respect to a certain simple parameterized family $\mathcal{H}_n^\epsilon$ of objective observables on $n$ qubits. The optima of the observables in $\mathcal{H}^\epsilon_n$ are $O(\epsilon)$ close to each other, in the sense that for any $H, H^\prime \in \mathcal{H}^\epsilon_n$, the ground state of $H$ is an $O(\epsilon)$ optimum of $H^\prime$. Precisely, we will prove the following.
\begin{theorem}[Zeroth-order lower bound]\label{thm:lower}
For any $n \geq 15$ and $\epsilon \leq 0.01n$, let $\mathcal{A}$ be any zeroth-order, $100\epsilon$-vicinity algorithm for the family $\mathcal{H}^\epsilon_n$ that makes $T$ queries to the oracle. Then, if $\Err(\mathcal{A},\mathcal{H}^\epsilon_n) \leq \epsilon$, it must hold that $T \geq \Omega\qty(\frac{n^3}{\epsilon^2})$ where the implicit factor is some fixed constant.
\end{theorem}

On the other hand, we prove that this same class of variational problems can be optimized substantially faster if the algorithm makes first-order queries to the oracle, as quantified in the following theorem. In fact, the algorithm that achieves this convergence rate is a simple stochastic gradient descent strategy. Hence, not only is the query complexity much better in this case, but the classical algorithm achieving this query complexity can be implemented efficiently. For comparison, we also obtain a zeroth-order upper bound for this class of problems using the algorithms of \cite{flaxman2005online} and \cite{agarwal2011stochastic}. Our first-order upper bound is given in the following theorem.

\begin{theorem}[First-order upper bound]\label{thm:upper}
For any $\epsilon \leq 0.01n$, there exists a first-order, $100\epsilon$-vicinity algorithm $\mathcal{A}$ for the family $\mathcal{H}^\epsilon_n$ that makes $O\qty(\frac{n^2}{\epsilon})$ queries and achieves an error $\Err(\mathcal{A},\mathcal{H}^\epsilon_n)\leq \epsilon$. Moreover, $\mathcal{A}$ is a simple stochastic gradient descent algorithm.
\end{theorem}

We also prove a very general lower bound for the case in which the algorithm may make $k\textsuperscript{th}$ order queries to the oracle for any $k$, and is not restricted to any particular domain of states.

\begin{theorem}[General lower bound]\label{thm:genLower}
For any $n\geq 15$ and $\epsilon \leq 0.01n$, suppose $\mathcal{A}$ is an algorithm that makes $T$ queries and satisfies $\Err(\mathcal{A},\mathcal{H}^\epsilon_n) \leq \epsilon$. Then $T \geq \Omega\qty(\frac{n^2}{\epsilon})$.
\end{theorem}
Since this lower bound is achieved (up to a possible constant factor) by the upper bound of SGD, we see that SGD is essentially optimal among all black-box strategies for optimizing $\mathcal{H}_n^\epsilon$.

\subsection{Defining $\mathcal{H}_n^\epsilon$}
The subset of objective observables we consider are perturbed around a very simple 1-local Hamiltonian.

\begin{definition}
Let $\delta \in \mathbb{R}$, and let $v \in \{-1,1\}^n$. Then we define
\begin{equation*}
H^{\delta}_v := - \sum_{i=1}^n \qty[ \sin(\frac{\pi}{4} + v_i \delta)X_i + \cos(\frac{\pi}{4} + v_i \delta)Z_i].
\end{equation*}
\end{definition}

Intuitively, for a fixed small parameter $\delta$, the set of $2^n$ observables $\{H^\delta_v\}_v$ are perturbed around $H^0 = -\frac{1}{\sqrt{2}} \sum_{i=1}^n \qty( X_i + Z_i)$. The parameter $\delta$ characterizes the strength of the perturbation, and the binary vector $v$ encodes the direction of the perturbation. It is straightforward to see that the ground state of $H^0$ is $\ket{\pi/4}^{\otimes n}$, where we have defined $\ket{\pi/4} := \cos(\pi/8) \ket{0} + \sin(\pi/8)\ket{1}$. Geometrically, the state $\ket{\pi/4}$ corresponds to the pure qubit state with polarization $\frac{1}{\sqrt{2}}(\hat{x}+\hat{z})$. In the remainder of this section, we record some facts about these Hamiltonians, and define some quantities. 

First, note that we may write $H^\delta_v = -\sum_{i=1}^n \hat{n}^{v_i \delta} \cdot \vec{\sigma}_i$ where $\hat{n}^{v_i \delta} = \qty(\sin(\frac{\pi}{4} + v_i \delta), 0, \cos(\frac{\pi}{4} + v_i \delta))$ and $\vec{\sigma}_i$ is the vector of Pauli operators acting on qubit $i$. We may now read off $\lambda_{\min}(H^\delta_v) = -n$, and the associated eigenvector is
\begin{equation*}
\ket{\psi_v^\delta} = \otimes_i \qty[\cos(\frac{\pi}{8} + \frac{v_i \delta}{2})\ket{0}_i + \sin(\frac{\pi}{8} + \frac{v_i \delta}{2})\ket{1}_i].
\end{equation*}

Next, we calculate the expectation value of $H^\delta_v$ with respect to any quantum state on $n$ qubits.

\begin{lemma}\label{lem:energy}
Suppose $\rho$ is a quantum state such that the polarization of $\rho_i$, the reduced state of $\rho$ on qubit $i$, is $\vec{r}_i$. Then $\tr\qty[H_v^{\delta} \rho] = -\sum_{i=1}^n \vec{r}_i \cdot \hat{n}^{\delta v_i}$. 
\end{lemma}
\begin{proof}
We have
\begin{align*}
\tr[H_v^\delta \rho] &= -\sum_{i=1}^n \tr[ (\hat{n}^{v_i \delta} \cdot \vec{\sigma}_i) \rho_i ] \\
&= - \frac{1}{2}\sum_{i=1}^n  \tr[ (\hat{n}^{v_i \delta} \cdot \vec{\sigma}_i) (I + \vec{r}_i\cdot \vec{\sigma}_i)] \\ 
&= -\frac{1}{2}\sum_{i=1}^n\tr[ (\hat{n}^{v_i \delta} \cdot \vec{\sigma}_i) (\vec{r}_i\cdot \vec{\sigma}_i)] \\ 
&= -\frac{1}{2}\sum_{i=1}^n \tr[ (\hat{n}^{v_i \delta} \cdot \vec{r}_i) I] \\
&= -\sum_{i=1}^n \hat{n}^{v_i \delta} \cdot \vec{r}_i.
\end{align*}
\end{proof}

Finally, we define the set $\mathcal{H}_n^\epsilon$ which we will prove the separation with respect to. To do so, we first define a bias parameter $\delta(\epsilon)$ associated with the precision parameter $\epsilon$.

\begin{definition}
For a given ``precision parameter'' $\epsilon$, define the associated ``bias parameter''
\begin{equation*}
\delta(\epsilon) := \sqrt{\frac{45\epsilon}{n}}.
\end{equation*}
\end{definition}

Now, we define $\mathcal{H}^\epsilon_n$ to be the set of such observables with bias parameter $\delta(\epsilon)$.

\begin{definition}\label{def:family}
$\mathcal{H}_n^{\epsilon} := \{ H^{\delta(\epsilon)}_{v} \, :\, \forall v\in \{-1,1\}^n \}$.
\end{definition}

For the remainder of the paper, we often hide the dependence of $\delta$ on $\epsilon$ for notational simplicity, and simply write $\delta$ where we implicitly mean $\delta(\epsilon)$. Note that our constraint $\epsilon \leq 0.01 n$ implies $\delta < 0.7$.

\subsection{Proof of Theorem \ref{thm:lower}: zeroth-order  lower bound for $\mathcal{H}_n^\epsilon$ in the vicinity of the optimum}
In this section, we prove Theorem \ref{thm:lower}. Our proof strategy for the lower bound is to reduce a statistical learning problem to the optimization problem, and then lower bound the number of oracle calls required to solve the learning problem. Precisely, we will take an appropriate subset  $\mathcal{M}^\epsilon_n \subset \mathcal{H}^\epsilon_n$, parameterized by some subset $\mathcal{V}$ of the $n$-dimensional hypercube $ \{-1,+1\}^n$. That is, we will have $\mathcal{M}^\epsilon_n = \{ H^{\delta(\epsilon)}_v\, :\, v\in \mathcal{V} \}$ where $\mathcal{V} \subset \{-1,1\}^n$ will be strategically chosen. We prove that, if there exists an algorithm $\mathcal{A}$ that satisfies $\Err(\mathcal{A},\mathcal{M}^\epsilon_n) \leq \epsilon$, then the same algorithm could be used to identify the hidden parameter $v\in \mathcal{V}$ associated with the objective observable $H^\delta_v\in \mathcal{M}^\epsilon_n$. By employing information theoretic methods, we will lower bound the number of oracle calls required to identify the parameter $v$, which in turn lower bounds the number of calls required to optimize to precision $\epsilon$.

Our proof in some parts adapts techniques from \cite{agarwal2009information} and \cite{jamieson2012query}, which lower bound the query cost of certain convex first-order and derivative-free optimization problems. These results in turn draw on methods from statistical minimax and learning theory. 

\subsubsection{Choosing a well-separated subset $\mathcal{M}^\epsilon_n \subset \mathcal{H}^\epsilon_n$}
We begin by defining, for fixed $\epsilon$, a subset $\mathcal{M}^\epsilon_n \subset \mathcal{H}_n^\epsilon$ of objective observables that are well-separated, in the sense that if a state is close to the optimal of $H_v^\delta \in \mathcal{M}^\epsilon_n$, then it must be far from the optimal of $H_{v^\prime}^\delta \in \mathcal{M}^\epsilon_n$ for any other parameter $v^\prime$. We make this precise below.

We make use of the following classical fact about packings of the hypercube (see for example \cite{guntuboyina2011lower} for a simple proof).

\begin{lemma}[Gilbert-Varshamov bound]
There exists a subset $\mathcal{V}$ of the $n$-dimensional hypercube $\{-1,1\}^n$ of size $| \mathcal{V} | \geq e^{n/8}$ such that, if $\Delta(v,v^\prime)$ denotes the Hamming distance between $v$ and $v^\prime$, 
\begin{equation*}
\Delta(v,v^\prime) \geq \frac{n}{4}
\end{equation*}
for all $v \neq v^\prime$ with $v, v^\prime \in \mathcal{V}$.
\end{lemma}

Fix $\mathcal{V}$ to be such a subset of $\{-1,1\}^n$, and define $\mathcal{M}^\epsilon_n := \{ H^\delta_v \, :\, v\in \mathcal{V}\}$. The Hamming distance provides a natural distance measure between points of the hypercube. We now define a notion of distance $d$ between objective observables $H_v^\delta$ and $H_{v^\prime}^\delta$. Intuitively, if $d(v, v^\prime)$ is large, then a state that is close to the optimal of $H_v^\delta$ cannot be close to the optimal of $H_{v^\prime}^\delta$. 

\begin{definition}\label{def:semimetric}
For $v, v^\prime \in \{-1,1\}^n$, we define the semimetric
\begin{equation*}
d(v,v^\prime) := \min_{\ket{\psi}} \qty[ \qty(\expval{H_v^\delta}{\psi} - \lambda_{\min} (H^\delta_v) ) + \qty(\expval{H_{v^\prime}^\delta}{\psi} - \lambda_{\min} (H^\delta_{v^\prime})) ]
\end{equation*}
where the minimization is over all normalized pure states on $n$ qubits.
\end{definition}

Note that $\lambda_{\min} (H^\delta_v)$ is simply $-n$, but we oftentimes write $\lambda_{\min} (H^\delta_v)$ for clarity. We now define a packing parameter $\beta$ which quantifies how packed the subset $\mathcal{V}$ is, with respect to the semimetric $d$. 

\begin{definition}
The packing parameter $\beta$ corresponding to the above subset $\mathcal{V}\subset \{-1,1\}^n$ and semimetric $d$ on the hypercube is defined to be
\begin{equation*}
\beta := \min_{v\neq v^\prime \in \mathcal{V}} d(v,v^\prime).
\end{equation*}
\end{definition}

We now have the following lemma.

\begin{lemma}\label{lem:packing}
Suppose that for some state $\ket{\psi}$ and parameter $v\in \mathcal{V}$, $\expval{H_v^\delta}{\psi} - \lambda_{\min} (H^\delta_v) \leq \beta / 3$. Then for all $v^\prime \neq v$ with $v^\prime \in \mathcal{V}$, $\expval{H_{v^\prime}^\delta}{\psi} - \lambda_{\min} (H^\delta_{v^\prime}) > \beta / 3$.
\end{lemma}
\begin{proof}
Suppose there exists some parameter $v^\prime \in \mathcal{V}$, $v^\prime \neq v$ for which $\expval{H_{v^\prime}^\delta}{\psi} -  \lambda_{\min} (H^\delta_{v^\prime}) \leq \beta / 3$. From Definition \ref{def:semimetric}, this implies that $d(v,v^\prime) \leq 2\beta / 3$, which contradicts the assumption that $\beta$ is the packing parameter. 
\end{proof}

We now show that any algorithm which optimizes the observables in the set $\mathcal{M}^\epsilon_n$  with error $\epsilon$ can be used to identify the parameter $v$ with high probability.

\begin{lemma}\label{lem:estimator}
Suppose that $\mathcal{A}$ is an algorithm such that $\Err(\mathcal{A},\mathcal{M}^\epsilon_n) \leq \beta / 9$.  Then, one may use the output of $\mathcal{A}$ to construct an estimator $\hat{v}$ such that, if the objective observable is $H_v^\delta$ for $v\in \mathcal{V}$, then $\Pr[\hat{v} =v] \geq 2/3$.  
\end{lemma}

\begin{proof}
By assumption, if the observable that is realized is $H^\delta_v$ for $v\in \mathcal{V}$, $\mathcal{A}$ outputs a description $\psi$ of a quantum state $\ket{\psi}$ such that 
\begin{equation*}
\E_{\psi}  \expval{H_v^\delta}{\psi} -  \lambda_{\min} (H^\delta_{v})  \leq \beta / 9.
\end{equation*}
Markov's inequality therefore implies
\begin{equation*}
\Pr_{\psi}[  \expval{H_v^\delta}{\psi} -  \lambda_{\min} (H^\delta_v)  \leq \beta/3 ] \geq 2 / 3.
\end{equation*}
Define the estimator $\hat{v}(\psi) := \argmin_{v^\prime \in \mathcal{V}} \expval{H_{v^\prime}^\delta}{\psi} - \lambda_{\min} (H^\delta_{v^\prime}) =  \argmin_{v^\prime \in \mathcal{V}} \expval{H_{v^\prime}^\delta}{\psi}$. Lemma \ref{lem:packing} implies that, if $ \expval{H_v^\delta}{\psi} - \lambda_{\min} (H^\delta_v)  \leq \beta/3$, this estimator returns $\hat{v} = v$ with probability one. Since this event occurs with probability at least $2/3$, the estimator returns $\hat{v} = v$ with probability at least $2/3$.
\end{proof}

We have shown that the ability to optimize $\mathcal{M}^\epsilon_n$ well implies the ability to identify the hidden parameter $v \in \mathcal{V}$ with high probability. We now compute the packing parameter $\beta$ for the family $\mathcal{M}^\epsilon_n$. 

\begin{lemma}\label{lem:packingCalculation}
For the subset $\mathcal{V}$, semimetric $d$, and packing parameter $\beta$ as defined above,
\begin{equation*}
\beta \geq \frac{n}{2}\qty( 1 - \cos(\delta) ) \geq \frac{n\delta^2}{5}.
\end{equation*}
\end{lemma}
\begin{proof}
Recall that for all $v,v^\prime \in \{-1,1\}^n$,
\begin{align*}
d(v,v^\prime) &= \min_{\ket{\psi}} \qty[ \qty(\expval{H_v^\delta}{\psi} - \lambda_{\min} (H^\delta_{v})) + \qty(\expval{H_{v^\prime}^\delta}{\psi} - \lambda_{\min} (H^\delta_{v^\prime})) ] \\
&= \min_{\ket{\psi}} \expval{(H_v^\delta + H_{v^\prime}^\delta)}{\psi} + 2n,
\end{align*}
where the minimization is over all normalized pure states on $n$ qubits.  Therefore, to compute $d(v,v^\prime)$, it suffices to compute the smallest eigenvalue of $H_v^\delta + H_{v^\prime}^\delta$. 

We may write 
\begin{align*}
H_v^\delta + H_{v^\prime}^\delta &= -\sum_{i: v_i = {v_i^\prime}}\qty[2 \sin(\frac{\pi}{4}+v_i \delta)X_i +  2\cos(\frac{\pi}{4}+v_i \delta)Z_i ] - \sum_{i: v_i \neq {v_i^\prime}}\qty[ \sqrt{2}\cos(\delta)X_i + \sqrt{2}\cos(\delta)Z_i ] \\
&=  -2 \sum_{i: v_i = {v_i^\prime}}\qty[ \sin(\frac{\pi}{4}+v_i \delta)X_i +  \cos(\frac{\pi}{4}+v_i \delta)Z_i ] - 2 \cos(\delta) \sum_{i: v_i \neq {v_i^\prime}} \qty[ \frac{1}{\sqrt{2}} X_i + \frac{1}{\sqrt{2}} Z_i ]
\end{align*}
where we used the trigonometric identities $\sqrt{2}\cos(\delta) = \cos(\pi/4 + \delta) +  \cos(\pi/4 - \delta) = \sin(\pi/4 + \delta) +  \sin(\pi/4 - \delta)$.  From this expression, it is clear that the smallest eigenvalue of $H_v^\delta + H_{v^\prime}^\delta$ is $-2(n-\Delta(v,v^\prime)) - 2 \cos(\delta)\Delta(v,v^\prime)$, from which it follows that $d(v,v^\prime) = 2\Delta(v,v^\prime)\qty(1-\cos(\delta))$. By construction, for all $v\neq v^\prime$ with $v,v^\prime \in \mathcal{V}$, we have $\Delta(v,v^\prime) \geq n/4$. It follows that $\beta \geq \frac{n}{2}(1-\cos(\delta))$.

The final inequality follows from the fact that $\cos(\delta)  \leq 1 - \frac{2 \delta^2}{5}$ for $\delta \leq 0.7$.
\end{proof}

\begin{lemma}\label{lem:reduction}
Any algorithm $\mathcal{A}$ for which $\Err(\mathcal{A},\mathcal{M}^\epsilon_n) \leq \epsilon$ can be used to construct an estimator $\hat{v}$ which correctly identifies the parameter $v$  of the realized observable $H^{\delta}_v \in \mathcal{M}^\epsilon_n$ with probability at least $2/3$.
\end{lemma}
\begin{proof}
By Lemma \ref{lem:packingCalculation}, the packing parameter is at least $\frac{n\delta^2}{5}$. Then by Lemma \ref{lem:estimator}, if we can optimize observables in the set $\mathcal{M}^\epsilon_n$ with expected error at most $\frac{1}{9} \frac{n\delta^2}{5} = \frac{n\delta^2}{45} = \epsilon$, we can identify $v$ with probability at least $2/3$.
\end{proof}

Our proof will proceed as follows. We restrict to the subset $\mathcal{M}^\epsilon_n \subset \mathcal{H}_n^\epsilon$ and prove a lower bound on the number of zeroth-order, $100\epsilon$-vicinity queries one must make in order to identify the hidden parameter $v$ associated with the realized objective observable $H^{\delta}_v \in \mathcal{M}^\epsilon_n$. By Lemma \ref{lem:reduction}, this number also lower bounds the number of such queries an algorithm $\mathcal{A}$ must make to satisfy $\Err(\mathcal{A},\mathcal{M}^\epsilon_n) \leq \epsilon$. Since $\mathcal{M}^\epsilon_n$ is a subset of $\mathcal{H}^\epsilon_n$, optimizing $\mathcal{M}^\epsilon_n$ is no harder than optimizing $\mathcal{H}^\epsilon_n$, and so this number also lower bounds the number of such queries needed to optimize $\mathcal{H}^\epsilon_n$ to precision $\epsilon$.

We next prove two simple lemmas we will need.

\begin{lemma}\label{lem:angles}
Suppose $\ket{\phi}$ is in the $\mu$-optimum of $H^\delta_v$, i.e. $\expval{H^\delta_v}{\phi} - \lambda_{\min} (H^\delta_{v})   \leq \mu$. Let $\vec{r}_i$ be the polarization of the reduced state of $\ket{\phi}$ on qubit $i$, and let $\alpha_i \in [0,\pi]$ be the (unoriented) angle between the vector $\vec{r}_i$ and the unit vector $\hat{n}^{\delta v_i}$. Then $\frac{1}{n}\sum_{i=1}^n \alpha_i^2 \leq \frac{10 \mu}{n}$ and $\frac{1}{n}\sum_{i=1}^n \alpha_i \leq \sqrt{\frac{10\mu}{n}}$.
\end{lemma}
\begin{proof}
From Lemma \ref{lem:energy} and rearranging terms,
\begin{equation*}
\frac{1}{n}\sum_{i=1}^n \vec{r}_i\cdot \hat{n}^{\delta v_i} \geq 1 - \frac{\mu}{n}.
\end{equation*}
Now, note that $\vec{r}_i\cdot \hat{n}^{\delta v_i} = |\vec{r}_i | \cos \alpha_i \leq |\vec{r}_i | \qty( 1 - \frac{\alpha_i^2}{10}) \leq 1 - \frac{\alpha_i^2}{10}$ for $\alpha_i \in [0,\pi]$. This gives us
\begin{equation*}
\frac{1}{n}\sum_{i=1}^n \alpha_i^2 \leq \frac{10\mu}{n}.
\end{equation*}
It immediately follows from Jensen's inequality that
\begin{equation*}
\frac{1}{n}\sum_{i=1}^n \alpha_i \leq \sqrt{\frac{10\mu}{n}}.
\end{equation*}
\end{proof}

\begin{lemma}\label{lem:vicinity}
Suppose $\ket{\phi}$ is in the $k\epsilon$-optimum of $H^{\delta}_{v^\prime}$ for some $k > 0$. Then $\ket{\phi}$ is in the $(k+30\sqrt{2k}+90)\epsilon$-optimum of $H^{\delta}_v$ for any $v\in \{-1,1\}^n$.
\end{lemma}
\begin{proof}
As in the previous lemma, let $\vec{r}_i$ denote the polarization of the reduced state on qubit $i$, and $\alpha_i$ denote the angle between $\vec{r}_i$ and $\hat{n}^{\delta v_i}$. We have
\begin{align*}
\expval{H^\delta_v}{\phi} - \lambda_{\min} (H^\delta_v) &= \expval{[H^\delta_{v^\prime} + (H^\delta_v- H^\delta_{v^\prime} )]}{\phi} - (-n) \\
&\leq k \epsilon + \expval{(H^\delta_v- H^\delta_{v^\prime} )}{\phi}.
\end{align*}
We now make the observation that 
\begin{equation*}
H^\delta_v - H^\delta_{v^\prime} = -2  \sin(\delta) \sum_{i \, : \, v_i \neq v^\prime_i} v_i \frac{X_i - Z_i}{\sqrt{2}}.
\end{equation*}
From this observation, it follows that
\begin{align*}
\expval{H^\delta_v}{\phi} - \lambda_{\min}(H^\delta_v) &\leq k \epsilon + 2\sin(\delta) \sum_{i=1}^n \abs{\expval{\frac{X_i - Z_i}{\sqrt{2}}}{\phi}} \\
&= k \epsilon + 2\sin(\delta) \sum_{i=1}^n \abs{ \vec{r}_i \cdot \frac{\hat{x} - \hat{z}}{\sqrt{2}} } \\
&\leq k \epsilon + 2 \sin(\delta) \sum_{i=1}^n \qty[ \alpha_i + \delta ] \\
&\leq k \epsilon + 2 \delta \sum_{i=1}^n [ \alpha_i + \delta] \\
&\leq k \epsilon + 2 \delta \sum_{i=1}^n \alpha_i + 2 n \delta^2  \\
& \leq k \epsilon +2\delta \sqrt{10 n k  \epsilon  } +  2n \delta^2  \\
& = (k + 30\sqrt{2k} + 90)\epsilon.
\end{align*}
\end{proof}
Here, the relation $\abs{ \vec{r}_i \cdot \frac{\hat{x} - \hat{z}}{\sqrt{2}} } \leq \alpha_i + \delta$ can be seen geometrically. We also used the definition $\delta^2 = \frac{45 \epsilon}{n}$. 

\subsubsection{Applying Fano's inequality}\label{sec:coin}
At this point, it remains to lower bound the number of zeroth-order calls to the oracle needed to correctly identify the unknown bias parameter $v\in \mathcal{V}$. The results from the previous section will then allow us to turn this into a lower bound for optimization. We will need the following well-known variant of Fano's inequality. For this result and other information-theoretic results used in this section, see (for example) \cite{cover1991elements}.

\begin{lemma}[Fano's inequality]\label{lem:Fano}
Suppose the random variable $V$ is uniformly distributed on the discrete set $\mathcal{V}$, and the variable $X$ may be correlated with $V$. Suppose $\mathcal{A}$ is an algorithm that attempts to identify $V$ given the variable $X$. Then the probability of error $p_e$ satisfies
\begin{equation*}
p_e \geq 1 - \frac{I(V;X) + 1}{\log |\mathcal{V}|}
\end{equation*} 
\end{lemma}
\noindent where $I(V;X)$ is the mutual information between $V$ and $X$. When we use this inequality in our proof, we will let $\mathcal{V}$ be the set of bias parameters $v\in \mathcal{V}$ associated with $\mathcal{M}^\epsilon_n$ defined in the previous section, and $X$ will be the set of queries to and outputs from the zeroth-order sampling oracle.

First, recall how the oracle behaves for zeroth-order queries. It selects a term in the Pauli expansion of the objective observable with probability proportional to the magnitude of the coefficient of that term. Consider the objective observable $H^\delta_v$. Note that the sum of coefficients of Pauli operators acting on qubit $i$ is $\sin(\pi/4 + v_i \delta) + \cos(\pi/4 + v_i \delta) = \sqrt{2}\cos(\delta)$ where we have used a standard trigonometric identity. Note that this quantity is independent of the parameter $v$. This means that, when we do a zeroth-order query of the oracle $\mathcal{O}_{H^\delta_v}$ encoding this Hamiltonian, the oracle is equally likely to select $X_i$ or $Z_i$ for measurement as it is $X_j$ or $Z_j$ for some other $j\neq i$. Thus, we may equivalently describe the oracle $\mathcal{O}_{H^\delta_v}$ as operating in the following manner. Note that the below algorithm is simply a specialization of the zeroth-order behavior of the sampling oracle (Definition \ref{def:zeroth}) to the particular objective observable $H^\delta_v$.

\noindent\fbox{%
    \parbox{\textwidth}{%
        \underline{\textbf{Zeroth order behavior of $\mathcal{O}_{H^\delta_v}$}}
        
        Upon input of a parameterization $\Theta$, parameter $\bm{\uptheta}$, and empty coordinate multiset $S = \varnothing$, 
            \begin{enumerate}
                \item Select an index $i\in [n]$ uniformly at random.
                \item Flip a coin with probability of heads $p = \frac{1}{\sqrt{2}\cos(\delta)} \sin(\pi/4 + v_i \delta) = \frac{1}{2}(1+v_i \tan(\delta))$.
                \item If heads, measure $-X_i$ w.r.t. the state $\ket{\bm{\uptheta}}$. If tails, measure $-Z_i$. \\
                \item Multiply the above measurement outcome by $E = \sqrt{2} n \cos(\delta)$ and output the result.
            \end{enumerate}
        
    }%
}
\begin{center}
\vspace{-1ex}Algorithm \hypertarget{alg5}{5}: zeroth-order behavior of $\mathcal{O}_{H^\delta_v}$. 
\end{center}

Let the parameter $v\in \mathcal{V}$ be uniformly distributed, and denote the associated random variable $V$. Suppose an algorithm makes $T$ zeroth-order queries to the oracle. Let $\xi_i$ be the input to the oracle in query $i$. Let $Y_i$ denote the output of query $i$. The algorithm may use information from steps one through $i$ to decide the input $\xi_{i+1}$ to query the oracle with on iteration $i+1$. Formally, we have the variables $\xi_1, Y_1, \xi_2, \dots,\xi_T, Y_T$, where $\xi_1$ (the algorithm's first guess) is independent of $V$, and $\xi_{i+1}$ is a deterministic or stochastic function of $\xi_1, Y_1, \dots, \xi_{i}, Y_{i}$. 
We begin with a simple lemma. Note that versions of this relation are well-known (e.g. \cite{agarwal2009information,raginsky2011information,jamieson2012query}).

\begin{lemma}\label{lem:chain}
$I(V ; (\xi_1,Y_1,\dots, \xi_T,Y_T)) \leq T \max_{\xi_1} I(V ; Y_1 | \xi_1)$.
\end{lemma}
\begin{proof}
\begin{align*}
I(V;(\xi_1,Y_1,\dots, \xi_T,Y_T)) &= \sum_{i=1}^{T} \qty[ \qty( I(V;\xi_i | \xi_1, Y_1, \dots, \xi_{i-1}, Y_{i-1}) + I(V;Y_i | \xi_1, Y_1, \dots, \xi_{i-1}, Y_{i-1},\xi_i)) ] \\
&= \sum_{i=1}^{T} I(V;Y_i | \xi_1, Y_1, \dots, \xi_{i-1}, Y_{i-1},\xi_i) \\
&= \sum_{i=1}^T (H(Y_i | \xi_1, Y_1, \dots, Y_{i-1},\xi_i) - H(Y_i | \xi_1, Y_1, \dots, Y_{i-1},\xi_i, V)) \\
&\leq \sum_{i=1}^T (H(Y_i | \xi_i) - H(Y_i | \xi_i, V)) \\
&= \sum_{i=1}^{T} I(V;Y_i | \xi_i) \\
&\leq T \max_{\xi_1} I(V;Y_1 | \xi_1)
\end{align*}
where in the first line we have used the chain rule for mutual information, in the second we used the fact that $\xi_i$ depends only on $(\xi_1, Y_1, \dots, \xi_{i-1}, Y_{i-1})$, in the third we used the definition of mutual information, and in the fourth we used subadditivity and the fact that $Y_i$ depends only on $\xi_i$ and $V$.
\end{proof}

With this  inequality in hand, we seek to upper bound $I(V ; Y_1 | \xi_1)$. To do so, we will write $I(V ; Y_1 | \xi_1)$ in terms of relative entropies. It will be helpful to introduce some additional notation. Let $\mathbb{P}$ be the distribution of $Y_1$ conditioned on $\xi_1$, $\mathbb{P}_v$ be the distribution of $\mathbb{P}$ conditioned on the hidden parameter being $v$, $\mathbb{P}^j$ be the distribution of $\mathbb{P}$ conditioned on the oracle selecting qubit $j$ for measurement in Step 1 of Algorithm \hyperlink{alg5}{5}, and $\mathbb{P}^j_{\pm 1}$ be the same distribution with the additional conditioning on $v_j = \pm 1$.

Letting $D(\cdot \| \cdot)$ denote the relative entropy of two distributions, we have for any $\xi_1$,
\begin{align*}
I(V ; Y_1 | \xi_1) &= \frac{1}{ |\mathcal{V}|}\sum_{v\in \mathcal{V}} D(\mathbb{P}_{v} \| \mathbb{P}) \\
&\leq \frac{1}{|\mathcal{V}|^2}\sum_{v,v^\prime \in \mathcal{V}} D(\mathbb{P}_{v} \| \mathbb{P}_{v^\prime}) \\
&\leq \frac{1}{ n |\mathcal{V}|^2}\sum_{v,v^\prime \in \mathcal{V}} \sum_{j=1}^n D(\mathbb{P}^j_{v_j} \| \mathbb{P}^j_{v^\prime_j}) \\
&\leq \max_{v,v^\prime\in \mathcal{V}} \frac{1}{n} \sum_{j=1}^n D(\mathbb{P}^j_{v_j} \| \mathbb{P}^j_{v^\prime_j})
\end{align*}
where the first line is a well-known expression for the mutual information, and the next two lines follow from convexity of the relative entropy. It remains to upper bound \\ $\max_{v,v^\prime} \frac{1}{n} \sum_{j=1}^n D(\mathbb{P}^j_{v_j} \| \mathbb{P}^j_{v^\prime_j})$.

\subsubsection{Upper bounding $\max_{v,v^\prime} \frac{1}{n} \sum_{j=1}^n D(\mathbb{P}^j_{v_j} \| \mathbb{P}^j_{v^\prime_j})$.}

Recall that since $\mathcal{A}$ only queries states in the $100\epsilon$-optimum of $\mathcal{H}^\epsilon_n$, then for any state $\ket{\bm{\uptheta}}$ that is queried, $\expval{H^{\delta}_v}{\bm{\uptheta}} \leq 650\epsilon$ by Lemma \ref{lem:vicinity}. As we have done before, let $\alpha_i$ denote the angle between $\vec{r}_i$ and $\hat{n}^{\delta v_i}$. By Lemma \ref{lem:angles}, we know that $\qty(\frac{1}{n}\sum_{i=1}^n \alpha_i)^2\leq \frac{1}{n}\sum_{i=1}^n \alpha_i^2 \leq \frac{6500 \epsilon}{n}$ for any state that is queried. 

Continuing on, we now calculate the distribution $\mathbb{P}^j_{v_j}$ in terms of previously defined parameters. Recall that $\mathbb{P}^j_{v_j}$ is a $\pm E$-valued Bernoulli distribution. Letting $\mathbb{P}^j_{v_j}[+E]$ denote the probability of obtaining $+E$, we find

\begin{align*}
\mathbb{P}^i_{v_i}[+E] &= \frac{1}{2}(1+v_i \tan(\delta)) \Pr[-X_i = +1] + \frac{1}{2}(1-v_i \tan(\delta)) \Pr[-Z_i = +1] \\
&= \frac{1}{2}(1+v_i \tan(\delta)) \frac{1}{2}(1-\vec{r}_i\cdot \hat{x}) + \frac{1}{2}(1-v_i \tan(\delta)) \frac{1}{2}(1-\vec{r}_i\cdot \hat{z}) \\
&= \frac{1}{2} - \frac{1}{4} \vec{r}_i \cdot (\hat{x} + \hat{z}) - \frac{v_i}{4} \tan(\delta) \vec{r}_i \cdot (\hat{x} - \hat{z}) \\
&= \frac{1}{2} - \frac{1}{2\sqrt{2}} \vec{r} \cdot\qty( \frac{\hat{x} + \hat{z}}{\sqrt{2}}+v_i \tan(\delta)\frac{\hat{x} - \hat{z}}{\sqrt{2}}).
\end{align*}

Similarly,
\begin{equation*}
\mathbb{P}^i_{v_i}[-E] = \frac{1}{2} + \frac{1}{2\sqrt{2}} \vec{r} \cdot\qty( \frac{\hat{x} + \hat{z}}{\sqrt{2}}+v_i \tan(\delta)\frac{\hat{x} - \hat{z}}{\sqrt{2}}).
\end{equation*} 

We are interested in bounding the relative entropy between $\mathbb{P}^j_{+1}$ and $\mathbb{P}^j_{-1}$. To this end, we first prove the following elementary lemma.

\begin{lemma}\label{lem:KL}
Suppose $Q$ and $R$ are two $\pm E$-valued Bernoulli distributions, with $Q[+E] = q$ and $R[+E] = r$. Then
\begin{equation*}
D(Q\| R) \leq \frac{1}{\ln 2} \frac{(q-r)^2}{r(1-r)}.
\end{equation*} 
\end{lemma}
\begin{proof}
\begin{align*}
D(Q \| R) &= q \log(\frac{q}{r}) + (1-q)\log(\frac{1-q}{1-r}) \\
&\leq \frac{q}{\ln 2}\qty(\frac{q}{r}-1) + \frac{1-q}{\ln 2}\qty(\frac{1-q}{1-r} - 1) \\
&= \frac{1}{\ln 2} \frac{(q-r)^2}{r(1-r)}
\end{align*}
where in the second line we used the inequality $\log x \leq \frac{1}{\ln 2} (x-1)$ with $x > 0$.
\end{proof}

Note that $\left| \frac{1}{2\sqrt{2}} \vec{r} \cdot\qty( \frac{\hat{x} + \hat{z}}{\sqrt{2}}+v_i \tan(\delta)\frac{\hat{x} - \hat{z}}{\sqrt{2}}) \right| \leq \frac{\sqrt{1+\tan^2 (\delta)}}{2\sqrt{2}} \leq 0.49$ for $\delta \leq 0.7$, which implies $0.01 \leq \mathbb{P}^i_{v_i}[+E] \leq 0.99$.  Further, we have 

\begin{align*}
\left| \mathbb{P}^i_{+1}[E] - \mathbb{P}^i_{-1}[E] \right|&= \left|\qty(\frac{1}{2} - \frac{1}{2\sqrt{2}} \vec{r}_i \cdot \frac{\hat{x} + \hat{z}}{\sqrt{2}} + \frac{1}{2 \sqrt{2}} \tan(\delta) \vec{r}_i \cdot \frac{\hat{x} - \hat{z}}{\sqrt{2}}) - \qty(\frac{1}{2} - \frac{1}{2\sqrt{2}} \vec{r}_i \cdot \frac{\hat{x} + \hat{z}}{\sqrt{2}} + \frac{-1}{2 \sqrt{2}} \tan(\delta) \vec{r}_i \cdot \frac{\hat{x} - \hat{z}}{\sqrt{2}}) \right|\\
&= \left| \frac{1}{\sqrt{2}} \tan(\delta) \vec{r}_i \cdot \frac{\hat{x}-\hat{z}}{\sqrt{2}}\right| \\
&\leq \frac{1}{\sqrt{2}} \tan(\delta)[\alpha_i + \delta ],
\end{align*}
where the last line follows from the same reasoning as in the proof of Lemma \ref{lem:vicinity}. Now, using Lemma \ref{lem:KL} we have
\begin{align*}
D(\mathbb{P}^i_{+1} \| \mathbb{P}^i_{-1}) &\leq \frac{1}{\ln 2}  \frac{  \tan^2(\delta) (\alpha_i + \delta)^2 / 2}{(0.99)(1-0.99)}  \\
&\leq \Theta(1) \delta^2 (\alpha_i + \delta)^2
\end{align*}
for $\delta \leq 0.7$, where $\Theta(1)$ denotes some fixed constant. An identical calculation shows that $D(\mathbb{P}^i_{-1} \| \mathbb{P}^i_{+1}) \leq \Theta(1) \delta^2 (\alpha_i + \delta)^2$. Finally, we have
\begin{align*}
\max_{v,v^\prime\in \mathcal{V}} \frac{1}{n}\sum_{i=1}^n D(\mathbb{P}^i_{v_i} \| \mathbb{P}^i_{v^\prime_i}) &\leq \max_{v,v^\prime\in \{-1,1\}^n } \frac{1}{n}\sum_{i=1}^n D(\mathbb{P}^i_{v_i} \| \mathbb{P}^i_{v^\prime_i}) \\
&= \frac{1}{n}\sum_{i=1}^n \max_{v_i, v^\prime_i \in \{-1,1\}} D(\mathbb{P}^i_{v_i}\| \mathbb{P}^i_{v^\prime_i}) \\
&\leq \frac{1}{n}\sum_{i=1}^n \Theta(1) \delta^2(\alpha_i + \delta)^2 \\
&\leq \Theta(1)\qty[ \delta^4 + 2\delta^3 \frac{1}{n}\sum_{j=1}^n \alpha_i + \delta^2 \frac{1}{n}\sum_{j=1}^n \alpha_i^2 ] \\
&\leq \Theta(1) \delta^4,
\end{align*}
where we have used  $\qty(\frac{1}{n}\sum_{i=1}^n \alpha_i)^2\leq \frac{1}{n}\sum_{i=1}^n \alpha_i^2 \leq \frac{6500 \epsilon}{n} = \Theta(1) \delta^2$.

\subsubsection{Completing the proof}\label{sec:completing}
Combining the above bound with Lemma \ref{lem:chain}, upon making $T$ zeroth-order queries to the oracle within the $100\epsilon$-optimum of $\mathcal{H}^\epsilon_n$, and obtaining outcomes $(Y_1, \dots, Y_T)$, the mutual information $I(V;(\xi_1, Y_1, \dots, \xi_T, Y_T))$ between the hidden vector $V$ and the inputs and outputs of the oracle is upper bounded by $O(T\delta^4)$. Lemma \ref{lem:Fano} implies that for any algorithm $\mathcal{A}$ which attempts to identify the bias vector $V$ given $(\xi_1, Y_1, \dots, \xi_T, Y_T)$, the error probability is lower bounded by $1-\frac{I(V;(\xi_1, Y_1, \dots, \xi_T, Y_T)) + 1}{\log |\mathcal{V}|}$. Recalling that $|\mathcal{V}| \geq e^{n/8}$, we have
\begin{align*}
p_e &\geq 1-\frac{I(V;(\xi_1, Y_1, \dots, \xi_T, Y_T)) + 1}{\log |\mathcal{V}|} \\
&\geq 1 - \frac{\Theta(1) T\delta^4 + 1}{\frac{1}{\ln 2} \frac{n}{8}} \geq 1 - \frac{\Theta(1) T\delta^4 + 1}{n/10}.
\end{align*}
Let $T_{1/3}$ be value of $T$ such that the final expression above is equal to $1/3$. A simple calculation shows $T_{1/3} = \frac{1}{\Theta(1) \delta^4}\qty(\frac{n}{15} - 1)$. In particular, for $n\geq 15$, $T_{1/3} \geq \Theta(1) \frac{n^3}{\epsilon^2}$. Note that, if an algorithm makes fewer than $T_{1/3}$ zeroth-order queries, the probability that it can correct identify the hidden parameter $V$ is less than $1/3$.

We have shown that for $n\geq 15$ and $\epsilon \leq 0.01n$, when constrained to the $100\epsilon$-optimum of $\mathcal{H}^\epsilon_n$, at least $\Omega\qty(\frac{n^3}{\epsilon^2})$ zeroth-order queries to the oracle are required to identify the bias parameter $v$ with probability of success at least $2/3$. Our previous reduction from learning to optimization then implies that at least this many samples are required to optimize observables in the set $\mathcal{H}_n^\epsilon$ with expected error at most $\epsilon$. We have therefore shown Theorem \ref{thm:lower}.

\subsection{Proof of Theorem \ref{thm:upper}: upper bound for optimizing $\mathcal{H}_n^\epsilon$}\label{sec:upperBoundProof}

We have shown that $\Omega\qty(\frac{n^3}{\epsilon^2})$ zeroth-order queries to the sampling oracle are required for a $100\epsilon$-vicinity algorithm to optimize the family $\mathcal{H}_n^\epsilon$ to precision $\epsilon$. In this section, we show that with a certain natural state parameterization and making only first-order queries to the sampling oracle, the family $\mathcal{H}_n^\epsilon$ can be optimized to precision $\epsilon$ with $O\qty(\frac{n^2}{\epsilon})$ queries by a $100\epsilon$-vicinity algorithm based on SGD.

We start by defining the variational ansatz that we will use in our first-order optimization procedure. We define the following $n$-parameter parameterization $\Theta$:
\begin{equation*}
\ket{\bm{\uptheta}} := \ket{\bm{\uptheta}_1, \dots, \bm{\uptheta}_n} := \otimes_{j=1}^n e^{-i (\bm{\uptheta}_j+\pi/4) Y_j / 2} \ket{0}^{\otimes n}.
\end{equation*}
This parameterization has a simple geometric interpretation: $\ket{\bm{\uptheta}}$ is the product state on $n$ qubits for which the polarization of qubit $j$ is $\sin(\pi/4+\bm{\uptheta}_j)\hat{x} + \cos(\pi/4+\bm{\uptheta}_j)\hat{z}$. Clearly this ansatz is natural for the family $\mathcal{H}^\epsilon_n$ in some sense.

Consider some objective observable $H^\delta_v \in \mathcal{H}_n^\epsilon$. From Lemma \ref{lem:energy}, we have that the induced objective function $f(\bm{\uptheta})$ is given by 
\begin{equation*}
f(\bm{\uptheta}) := \expval{H^\delta_v}{\bm{\uptheta}} = -\sum_{i=1}^n \cos(\bm{\uptheta}_i -  \delta v_i).
\end{equation*}

Let $\mathcal{B}_\infty(\delta) \subset \mathbb{R}^n$ denote the $\infty$-ball of radius $\delta$ centered at the origin. Precisely, $\mathcal{B}_\infty (\delta) = \{ \bm{\uptheta} \, :\, \| \bm{\uptheta} \|_\infty \leq \delta \}$. Note that the ground state of $H^\delta_v$ is the state $\ket{\delta v_1, \delta v_2, \dots, \delta v_n}$, and hence corresponds to a parameter inside the set $\mathcal{B}_\infty(\delta)$ for any bias vector $v$. Furthermore, the set of states associated with $\mathcal{B}_\infty(\delta)$ is contained in the $100\epsilon$-optimum of $\mathcal{H}^\epsilon_n$. We state this fact as the following lemma.

\begin{lemma}
The set of states associated with $\mathcal{B}_\infty(\delta)$ is contained in the $100\epsilon$-optimum of $\mathcal{H}^\epsilon_n$.
\end{lemma}
\begin{proof}
For any $H^\delta_v \in \mathcal{H}^\epsilon_n$ and $\bm{\uptheta} \in \mathcal{B}_\infty(\delta)$, we have
\begin{equation*}
\expval{H^\delta_v}{\bm{\uptheta}} - \lambda_{\min} (H^\delta_{v}) = \expval{H^\delta_v}{\bm{\uptheta}} - (-n) \leq n(1-\cos(2\delta)) \leq 2n\delta^2 = 90\epsilon
\end{equation*}
where we used $\cos(x)\geq 1-x^2 / 2$. 
\end{proof}

 We now will show that $f(\bm{\uptheta})$ is $0.1$-strongly convex on $\mathcal{B}_\infty(\delta)$ w.r.t. the Euclidean norm. To do so, we compute its Hessian matrices $\nabla^2 f(\bm{\uptheta})$. We have $(\nabla^2 f(\bm{\uptheta}))_{ij} = \pdv[2]{f}{\bm{\uptheta}_i}{\bm{\uptheta}_j} (\bm{\uptheta}) = 0$ for $i\neq j$, and $(\nabla^2 f(\bm{\uptheta}))_{ii} = \pdv[2]{f}{\bm{\uptheta}_i} (\bm{\uptheta}) = \cos(\bm{\uptheta}_i - \delta v_i)$. Since $\bm{\uptheta}_i \in [-\delta , \delta]$, it must hold that $(\nabla^2 f(\bm{\uptheta}))_{ii} \geq \cos(2\delta) \geq 0.1$ where we used our assumption $\delta < 0.7$ for the last inequality. Since all eigenvalues of $\nabla^2 f(\bm{\uptheta})$ for $\bm{\uptheta}\in \mathcal{B}$ are at least $0.1$, $f$ is $0.1$-strongly convex on $\mathcal{B}_\infty(\delta)$.

We now calculate $\vec{\Gamma}$ for this particular parameterization and some objective observable $H^\delta_v$. Expanding the gradient as in Section \ref{sec:black-box} (see also Table \ref{tab:notation}),

\begin{align*}
\nabla f(\bm{\uptheta}) = -\sum_{j=1}^n \sum_{k=1}^n  \expval{\frac{i}{2}\comm{U_{(j+1):n} Y_j U_{(j+1):n}^\dagger}{\sin(\frac{\pi}{4}+\delta v_k)X_k + \cos(\frac{\pi}{4}+\delta v_k)Z_k}}{\bm{\uptheta}} \hat{e}_j 
\end{align*}

where, as usual, $U_{(j+1):n} := e^{-i \bm{\uptheta}_n Y_n / 2} \cdots e^{-i\bm{\uptheta}_{j+1} Y_{j+1} / 2}$. We now  remove terms which are trivially zero because the commutator involves operators which act nontrivially on disjoint qubits. In particular, since $U_{(j+1):n} Y_j U_{(j+1):n}^\dagger = Y_j$ in this case, then clearly $\text{qubits}(U_{(j+1):n} Y_j U_{(j+1):n}^\dagger) = \{j\}$. Dropping such terms in the expansion,

\begin{align*}
\nabla f(\bm{\uptheta}) = -\sum_{j=1}^n  \expval{\frac{i}{2}\comm{U_{(j+1);n} Y_j U_{(j+1);n}^\dagger}{\sin(\frac{\pi}{4}+\delta v_j)X_j + \cos(\frac{\pi}{4}+\delta v_j)Z_j}}{\bm{\uptheta}} \hat{e}_j.
\end{align*}

Recall that $\Gamma_j$ is the sum of the magnitudes of the coefficients of the above expansion for component $j$. In particular, we have $\Gamma_j = \sin(\frac{\pi}{4} + \delta v_j) + \cos(\frac{\pi}{4}+\delta v_j) =  \sqrt{2} \cos(\delta) = \Theta(1)$. Now, Lemma \ref{lem:SGDSCv} implies that projected SGD, using the feasible set $B_\infty(\delta)$, outputs a parameter $\bar{\bm{\uptheta}}$ such that $\E f(\bar{\bm{\uptheta}}) - \lambda_{\min}(H) \leq \epsilon$ for all $H \in \mathcal{H}^\epsilon_n$ using $O(\| \vec{\Gamma} \|_1^2 / \lambda_2 \epsilon) = O(n^2 / \epsilon)$ queries, where $\lambda_2$ is the strong convexity parameter (w.r.t. Euclidean norm) which is $\Theta(1)$ in our case.

As a sidenote, we point out that simply running some version of SGD with no projections (using $\mathbb{R}^n$ as the feasible set) would likely perform very well for this problem, since all local optima are also global minima in this case (even though the objective function is nonconvex on $\mathbb{R}^n$).

\subsection{Proof of Theorem \ref{thm:genLower}: general query lower bound for optimizing $\mathcal{H}^\epsilon_n$}\label{sec:generalLower}
Using a very similar argument to that of the proof of Theorem \ref{thm:lower}, we may lower bound the number of calls to $\mathcal{O}_H$  required to optimize any objective observable in the family $\mathcal{H}^\epsilon_n$ with expected error at most $\epsilon$. In the setting of Theorem \ref{thm:lower}, the algorithm was restricted to querying the oracle with states in the $100\epsilon$-optimum of $\mathcal{H}^\epsilon_n$. In this section, the algorithm is allowed to query the oracle with states which may be outside this domain.  We also allow the algorithm to make queries of any order, instead of just zeroth-order. As before, we actually prove a lower bound for the strictly easier problem of optimizing the subset $\mathcal{M}^\epsilon_n \subset \mathcal{H}^\epsilon_n$. 

We will essentially bound the amount of information contained in a single oracle query for any order derivative and for any state. As usual, let $\Theta$ denote the parameterization given by $\ket{\bm{\uptheta}} = e^{-iA_p \bm{\uptheta}_p / 2}\cdots e^{-iA_1 \bm{\uptheta}_1 / 2}\ket{\Psi}$. Recall from Section \ref{sec:higher} that, assuming w.l.o.g. that $j_1 \leq \cdots \leq j_r$, the expansion of $\frac{\partial^r f}{\partial \theta_{j_1}\cdots \partial \theta_{j_r}} (\bm{\uptheta})$ in terms of nested commutators of conjugated Pauli operators is

\begin{equation*}
\frac{\partial^r f}{\partial \theta_{j_1}\cdots \partial \theta_{j_r}} (\bm{\uptheta}) = \qty(\frac{i}{2})^r \sum_{k_1=1}^{n_{j_1}} \cdots \sum_{k_r = 1}^{n_{j_r}} \left(\prod_{i=1}^r \beta_{k_i}^{(j_i)}\right) \expval{\left[ \tilde{Q}^{(j_1)}_{k_1},\left[ \dots , \left[ \tilde{Q}^{(j_r)}_{k_r}, \sum_{l=1}^m \alpha_l P_l \right]\dots \right] \right]}{\bm{\uptheta}}
\end{equation*}

where $A_j = \sum_{k=1}^{n_j} \beta^{(j)}_k Q^{(j)}_k$, $H = \sum_{l = 1}^m \alpha_l P_l$, and the notation $\tilde{Q}$ is defined in Section \ref{sec:higher}. Specialize to the case in which $H = H^\delta_n \in \mathcal{M}^\epsilon_n$. Then, after removing nested commutators which are trivially zero, we may write
\begin{equation*}
\frac{\partial^r f}{\partial \theta_{j_1}\cdots \partial \theta_{j_r}} (\bm{\uptheta}) = \sum_{k_1=1}^{n_{j_1}} \cdots \sum_{k_r = 1}^{n_{j_r}} \sum_{l=1}^n \zeta_{k_1,\dots,k_r, l} \expval{\left[ \tilde{Q}^{(j_1)}_{k_1},\left[ \dots , \left[ \tilde{Q}^{(j_r)}_{k_r}, \sin(\pi/4 + \delta v_l) X_l + \cos(\pi/4 + \delta v_l) Z_l\right]\dots \right] \right]}{\bm{\uptheta}}
\end{equation*}

for some coefficients $\zeta_{k_1,\dots,k_r, l}$ that are independent of $v$. Recall how $\mathcal{O}_{H^\delta_v}$ operates upon a query for the derivative $\frac{\partial^r f}{\partial \theta_{j_1}\cdots \partial \theta_{j_r}} (\bm{\uptheta})$. After doing a Pauli decomposition of the original nested commutator expression for the derivative and removing terms that are trivially zero, it samples a term with probability proportional to the magnitude of the coefficient of that term, and then obtains an unbiased estimator for that term using the procedure outlined in Section \ref{sec:higher}. Hence, we may equivalently describe the behavior of the oracle upon an $r\textsuperscript{th}$-order query of some state $\bm{\uptheta}$ as follows.

\noindent\fbox{%
    \parbox{\textwidth}{%
        \underline{\textbf{$r\textsuperscript{th}$-order behavior of $\mathcal{O}_{H^\delta_v}$.}}
        
        Upon input of a parameterization $\Theta$, parameter $\bm{\uptheta}$, and coordinate multiset $S = \{ j\}$,
            \begin{enumerate}
                \item Select indices $(k_1,\dots,k_r,l)$ with probability proportional to $|\zeta_{k_1,\dots,k_r,l}|$.
                \item Flip a coin with probability of heads $p = \frac{1}{\sqrt{2}\cos(\delta)} \sin(\pi/4 + v_l \delta) = \frac{1}{2}(1+v_l \tan(\delta))$.
                \item If heads, estimate $\expval{\left[ \tilde{Q}^{(j_1)}_{k_1},\left[ \dots , \left[ \tilde{Q}^{(j_r)}_{k_r}, X_l\right]\dots \right] \right]}{\bm{\uptheta}}$ with a single-measurement generalized Hadamard test using the procedure of Section \ref{sec:higher}. If tails, estimate $\expval{\left[ \tilde{Q}^{(j_1)}_{k_1},\left[ \dots , \left[ \tilde{Q}^{(j_r)}_{k_r}, Z_l\right]\dots \right] \right]}{\bm{\uptheta}}$.
                \item Multiply the result of Step 3 by the appropriate normalization factor and output the result.
            \end{enumerate}
    }%
}
\begin{center}
\vspace{-1ex}Algorithm 6: $r\textsuperscript{th}$-order behavior of $\mathcal{O}_{H^\delta_v}$. 
\end{center}

The crucial point is that the sampling oracle cannot reveal any more information about the hidden parameter $v$ than the outcome of the internal coin flip in Step 2 of the above box. This is because, since only Step 2 in the above box depends on the hidden parameter $v$, the algorithm can simulate the oracle if it has knowledge of the outcome of the internal coin flip. More formally, we have the following lemma. 
\begin{lemma}
Let $V$ be the hidden parameter, $\xi$ be the input to the oracle, $W$ be the outcome of the internal coin flip, and $Y$ be the output of the oracle. Then $I(V;Y | \xi) \leq I(V;W | \xi)$.
\end{lemma}
\begin{proof}
Note from the above box that the coin flip of Step 2 is the only part of the black box's internal procedure that depends on $V$; the output $Y$ is simply a stochastic function of $W$. Hence the variables $V \rightarrow W \rightarrow Y$ form a Markov chain, and the claim follows from the data processing inequality.
\end{proof}

We now use this observation along with a similar argument to that of Section \ref{sec:coin} to derive the desired lower bound. As before, we have $I(V;(\xi_1,Y_1,\dots, \xi_T,Y_T)) \leq T \max_{\xi_1} I(V; Y_1 | \xi_1)$, so we will seek to upper bound $I(V;Y_1 | \xi_1)$.  Let $\mathbb{Q}$ denote the distribution of $W_1$. Let $\mathbb{Q}^l$ denote the distribution of $W_1$ conditioned on the oracle selecting $L=l$ in Step 1 of Algorithm 6. Let $\mathbb{Q}_v$ denote the distribution of $W_1$  conditioned on the hidden parameter $V=v$. Let $\mathbb{Q}_{v_l}^l$ denote the distribution of $W_1$  conditioned on $V_l=v_l$ and the oracle selecting $L=l$ in Step 1 of Algorithm 6. Using convexity of relative entropy multiple times, we have

\begin{align*}
I(V;Y_1 | \xi_1) &\leq I(V;W_1 | \xi_1) \\
&= \frac{1}{|\mathcal{V}|} \sum_{v\in \mathcal{V}} D(\mathbb{Q}_v \| \mathbb{Q}) \\ 
&\leq \frac{1}{|\mathcal{V}|^2} \sum_{v,v^\prime \in \mathcal{V}} D(\mathbb{Q}_v \| \mathbb{Q}_{v^\prime}) \\
&\leq \frac{1}{|\mathcal{V}|^2} \sum_{v,v^\prime \in \mathcal{V}} \E_{L} D(\mathbb{Q}^L_{v_L} \| \mathbb{Q}^L_{v_L^\prime}) \\
&\leq \max_{v,v^\prime \in \mathcal{V}} \E_L D(\mathbb{Q}^L_{v_L} \| \mathbb{Q}^L_{v_L^\prime}) \\
&\leq \max_{v,v^\prime \in \mathcal{V}} \max_{l} D(\mathbb{Q}^l_{v_l} \| \mathbb{Q}^l_{v^\prime_l})
\end{align*}
where we have used the fact that $\mathbb{Q}_v = \mathbb{E}_L \mathbb{Q}^L_{v_L}$, where the expectation value is over the choice of parameter $L$ made by the black box.

It remains to bound the final expression above. Clearly $D(\mathbb{Q}^l_{+1} \| \mathbb{Q}^l_{+1}) = D(\mathbb{Q}^l_{-1} \| \mathbb{Q}^l_{-1}) = 0$.  Now we calculate $D(\mathbb{Q}^l_{+1} \| \mathbb{Q}^l_{-1})$. Recall that for the distribution $\mathbb{Q}^l_{v_l}$, the probability of ``heads'' is $\frac{1}{2}(1+v_l \tan(\delta))$. By nearly identical arguments to those in the proof of Theorem \ref{thm:lower}, it then follows immediately from Lemma \ref{lem:KL} that there exists some constant $c$ such that $D(\mathbb{Q}^l_{+1} \| \mathbb{Q}^l_{-1}) \leq c \delta^2$ for $\delta < 0.7$. Similarly, $D(\mathbb{Q}^l_{-1} \| \mathbb{Q}^l_{+1}) \leq c \delta^2$. 

At this point, we may follow a virtually identical argument to that in Section \ref{sec:completing} to find that, for $n\geq 15$ and $\epsilon \leq 0.01 n$, at least $\Omega\qty(\frac{n^2}{\epsilon})$ oracle queries are required to identify the hidden bias parameter $v$ with probably at least $2/3$. Hence, at least $\Omega\qty(\frac{n^2}{\epsilon})$ oracle queries are required to optimize $\mathcal{H}_n^\epsilon$ with worst-case expected error at most $\epsilon$.

Since this lower bound has a matching upper bound via first-order oracle queries and SGD (up to constant factors), we see that SGD is in fact essentially optimal among all black-box strategies for optimizing the family $\mathcal{H}_n^\epsilon$.

\section{Conclusion and open questions}\label{sec:conclusion}
We have introduced a natural black-box setting for variational algorithms, which can be straightforwardly implemented in practice. With respect to this setting, we derived rigorous upper bounds on the query cost of variational algorithms, in the setting where the induced objective function is convex within a convex feasible set. These bounds depended on the precision, dimension of parameter space, factors from the objective observable and pulse generators, and strong convexity parameters. We derived bounds both for algorithms running SGD in a Euclidean space, and for algorithms running SMD in an $l_1$ space. For some settings of parameters SGD has stronger upper bounds, and for other settings of parameters SMD has stronger upper bounds. For the toy problem we analyze, SGD outperforms SMD by a factor that is merely logarithmic in the number of parameters. It is an interesting open question to understand which geometry is most natural for variational algorithms in practice. 

We also introduced a simple class of objective observables $\mathcal{H}_n^\epsilon$ on $n$ qubits, and proved a separation between the query cost of optimizing these observables in the vicinity of the optimum in the cases of zeroth-order (objective function measurements) versus first-order (analytic gradient measurements) optimization. We showed that, for this class of observables, a simple stochastic gradient descent strategy could outperform any possible variational algorithm (with any choice of ansatz) that only receives zeroth-order information from the oracle. We view these results as evidence that taking analytic gradient measurements in variational algorithms and using the measurement results to run a stochastic first-order optimization algorithm could be advantageous as compared to derivative-free strategies in some cases. 

It would be interesting to understand the behavior of the objective function $f(\bm{\uptheta})$ near a local minimum for problems and variational ansatzes which appear in practice. In particular, it would be interesting to understand how the strong convexity of $f(\bm{\uptheta})$ typically behaves near a local minimum. Without a strong convexity guarantee, stochastic descent methods typically have query upper bounds scaling with the precision like $O(1/\epsilon^2)$. However, given a promise of $\lambda$-strong convexity, the cost is typically $O(1/\lambda \epsilon)$.  For the toy model $\mathcal{H}_n^\epsilon$ that we analyzed, we showed that with an appropriate choice of ansatz, the problem was $\Theta(1)$-strongly convex with respect to the 2-norm.  As a result of this property, we were able to obtain a $O(1/\epsilon)$ query upper bound for optimizing this family with SGD. We apparently were able to exploit strong convexity by making a prudent choice of variational ansatz for the problem class at hand.

This situation may be viewed as the opposite of that studied in \cite{mcclean2018barren}, which essentially considered a situation in which the variational ansatz looks random. In this situation, the gradient of the objective function is highly concentrated around zero. One way to view the difference in our models is that in our paper there are $n$ independent (i.e.~commuting) degrees of freedom while in \cite{mcclean2018barren} different terms in the Hamiltonian and pulses have the commutation relations that we would expect from Haar-random projectors.  Our model could be seen as justified by the common intuition in many-body physics that local unitaries applied to the ground state create quasiparticles, and that in an $n$-qubit system  $O(n)$ independent quasiparticles are possible.  Their model, on the other hand, could be justified by the assumption that the variational ansatz is far from a local minimum and so the pulses act like random local unitaries.
It would be interesting to understand which of these scenarios is more realistic in practice. In particular, one might hope that theoretically motivated ansatzes, such as the unitary-coupled-cluster ansatz in quantum chemistry, could possess properties near an optimum (such as strong convexity) that make them more amenable to efficient optimization.

As another open problem, recall that for the toy problem $\mathcal{H}^\epsilon_n$ we proved that taking analytic $k\textsuperscript{th}$-order derivative measurements for $k \geq 2$  provides no benefit over taking zeroth- and first-order measurements. However, the observables in the family $\mathcal{H}^\epsilon_n$ are extremely simple, being merely 1-local and having unentangled ground states. It seems plausible that taking second (or higher) order measurements could be beneficial for more complicated problems. It would be interesting to understand how higher-order measurements could improve convergence in such cases.

Finally, another point that we left unaddressed is the issue of noise. It would be interesting to study how to take analytic gradient measurements in the presence of noise, and what impact this has on the convergence rate of stochastic optimization methods. In particular, these methods are quite robust against unbiased noise, but their effectiveness in the presence of biased noise is less understood.

\section*{Acknowledgements}
We thank an anonymous reviewer for helpful suggestions, and for pointing out the good practical effectiveness of derivative-free trust region and surrogate methods.  We thank Xiaodi Wu for useful discussions.
JN and AWH were funded by ARO contract W911NF-17-1-0433 and NSF grants CCF-1729369 and PHY-1818914. 
AWH was also funded by NSF grant CCF-1452616 and  the MIT-IBM Watson AI Lab under the project {\it Machine Learning in Hilbert space}.

\bibliographystyle{alpha}
\bibliography{gradient}

\appendix
\section{Background on stochastic gradient and mirror descent}\label{appendix}
In this section, we review some relevant preliminaries pertaining to convex optimization and stochastic descent algorithms. Much of the material in this section follows the review \cite{bubeck2015convex}.

\subsection{Gradient descent}

We first describe the \emph{projected gradient descent} scheme for minimizing a convex differentiable function $f$ on some compact convex subset $\mathcal{X}\subset \mathbb{R}^n$. Starting from some initial point $\vb{x}_1 \in \mathcal{X}$, iterate the following procedure:
\begin{equation*}\label{eq:gradStep}
\vb{x}_{t+1} = \Pi_{\mathcal{X}} (\vb{x}_t - \eta_t \nabla f(\vb{x}_t))
\end{equation*}
where $\eta_t > 0$ is the stepsize at iteration $t$, and $\Pi_{\mathcal{X}}$ is the Euclidean projection onto $\mathcal{X}$, $\Pi_{\mathcal{X}}(\vb{x}) = \argmin_{\vb{y}\in \mathcal{X}} \| \vb{x}- \vb{y}\|_2$. The intuition for this strategy is clear: the vector $-\nabla f(\vb{x}_t)$ points in the direction of steepest decrease of $f$ at $\vb{x}_t$, and in each iteration we take a step of size $\eta_t$ in this direction and then project back into $\mathcal{X}$. It is sometimes helpful to think of gradient descent in an alternative, \emph{proximal} picture. Namely, Eq. \ref{eq:gradStep} is equivalent to
\begin{equation*}
\vb{x}_{t+1} = \argmin_{\vb{x} \in \mathcal{X}}\qty[ f(\vb{x}_t) + \nabla f(\vb{x}_t)^{\top} (\vb{x}-\vb{x}_t) + \frac{1}{2 \eta_t} \| \vb{x} - \vb{x}_t \|_2^2 ].
\end{equation*}

Intuitively, the point $\vb{x}_{t+1}$ is chosen to minimize $f(\vb{x}_t) + \nabla f(\vb{x}_t)^{\top} (\vb{x}-\vb{x}_t)$, a linearization of $f$ around $\vb{x}_t$, while not making the regularization term $\frac{1}{2 \eta_t} \| \vb{x} - \vb{x}_t \|_2^2$ too big.

The following result about projected gradient descent is well-known. Recall that a differentiable function $f$ is $L$-Lipschitz with respect to $\| \cdot \|$ if $\| \nabla f(\vb{x}) \|_* \leq L$ for all $\vb{x}\in\mathcal{X}$, where $\|\cdot \|_*$ denotes the dual norm.  

\begin{theorem}
If the convex function $f$ is $L$-Lipschitz w.r.t. the Euclidean norm, and $\mathcal{X}$ is contained in a Euclidean ball of radius $R_2$, then  projected gradient descent with stepsize $\eta = \frac{R_2}{L\sqrt{T}}$  satisfies 
\begin{equation*}
f\qty(\frac{1}{T}\sum_{s=1}^{T} \vb{x}_s ) - f(\vb{x}^*) \leq \frac{R_2 L }{\sqrt{T}}
\end{equation*}
where $\vb{x}^*$ is a minimizer of $f$ on $\mathcal{X}$.
\end{theorem}
Note that this implies that $\frac{R_2^2 L^2}{\epsilon^2}$ iterations are sufficient for some desired precision $\epsilon$. We now define strong convexity.

\begin{definition}[Strong convexity]
The function $f : \mathcal{X} \rightarrow \mathbb{R}$ is \emph{$\lambda$-strongly convex} with respect to arbitrary norm $\| \cdot \|$ for $\lambda > 0$ if for all $\vb{x},\vb{y} \in \mathcal{X}$, $f(\vb{y}) \geq f(\vb{x}) + \nabla f(\vb{x})^{T} (\vb{y}-\vb{x})  + \frac{\lambda}{2} \| \vb{x} - \vb{y} \|^2$.
\end{definition}
Note that if $f$ is twice differentiable, then $f$ is $\lambda$-strongly convex with respect to $\| \cdot \|_2$ if and only if the eigenvalues of the Hessians of $f$ are all at least $\lambda$. Recall that $f$ is convex if and only if the Hessians of $f$ are all positive semidefinite. In general, $f$ is $\lambda$-strongly convex w.r.t. arbitrary norm $\| \cdot \|$ if and only if  $\forall \vb{x}\in \mathcal{X}, \vb{h} \in \mathbb{R}^p$, it holds that $\vb{h}^{\top}\nabla^2 f(\vb{x}) \vb{h} \geq \lambda \| \vb{h} \|^2$.

The following result about projected gradient descent for strongly convex functions is known.

\begin{theorem}\label{thm:stronglyConvex}
Let $f$ be $\lambda_2$-strongly convex and $L$-Lipschitz on $\mathcal{X}$, w.r.t. the Euclidean norm. Then projected gradient descent with $\eta_s = \frac{2}{\lambda_2 (s+1)}$ satisfies
\begin{equation*}
f\qty(\sum_{s=1}^{T} \frac{2s}{T(T+1)} \vb{x}_s) - f(\vb{x}^*) \leq \frac{2L^2}{\lambda_2 (T+1)}.
\end{equation*}
\end{theorem}

\noindent Note that this result implies that $O\qty(\frac{L^2}{\lambda_2 \epsilon})$ iterations are sufficient to optimize $f$ to error $\epsilon$.

It turns out that if one does gradient descent with noisy, unbiased estimates of the gradient $\hat{\vb{g}}(\vb{x})$ instead of the true gradient $\nabla f(\vb{x})$, the above results are qualitatively unchanged. We refer to projected gradient descent with stochastic gradient estimates as \emph{stochastic gradient descent} (SGD). We now state some results formally. 

\begin{theorem}
Let $f$ be convex on $\mathcal{X}$, which is contained in a Euclidean ball of radius $R_2$. Assume we have access to a stochastic gradient oracle, which upon input of $\vb{x}\in \mathcal{X}$, returns a random vector $\hat{\vb{g}}(\vb{x})$ such that $\mathbb{E} \hat{\vb{g}}(\vb{x}) = \nabla f(\vb{x})$ and $\mathbb{E} \| \hat{\vb{g}}(\vb{x}) \|_2^2 \leq G_2^2$. Then SGD with $\eta = \frac{R_2}{G_2\sqrt{T}}$  satisfies 
\begin{equation*}
f\qty(\frac{1}{T}\sum_{s=1}^{T} \vb{x}_s ) - f(\vb{x}^*) \leq \frac{R_2 G_2 }{\sqrt{T}}.
\end{equation*}
\end{theorem}

\begin{theorem}
Let $f$ be $\lambda_2$-strongly convex on $\mathcal{X}$ w.r.t. the Euclidean norm. Assume we have access to a stochastic gradient oracle, which upon input of $\vb{x}\in \mathcal{X}$, returns a random vector $\hat{\vb{g}}(\vb{x})$ such that $\mathbb{E} \hat{\vb{g}}(\vb{x}) = \nabla f(\vb{x})$ and $\mathbb{E} \| \hat{\vb{g}}(\vb{x}) \|^2_2 \leq G_2^2$. Then SGD with step sizes $\eta_s = \frac{2}{\lambda_2 (s+1)}$ satisfies 
\begin{equation*}
f\qty( \sum_{s=1}^{T} \frac{2s}{T(T+1)} \vb{x}_s) - f(\vb{x}^*) \leq \frac{2G_2^2}{\lambda_2(T+1)}.
\end{equation*}
\end{theorem}

\subsection{Mirror descent}\label{sec:mirrorDescent}
A reflection on gradient descent shows that the gradient descent procedure defined above in fact only makes sense when we are working in Euclidean space. For example, an iteration of gradient descent (Eq. \ref{eq:gradStep}) involves adding the vectors $\vb{x}_t$ and $\eta_t \nabla f(\vb{x}_t)$. When the problem is defined in Euclidean space, $\nabla f(\vb{x}_t)$ may be considered as living in the same space by the Riesz representation theorem. However, if (for example) the objective function $f$ is defined on an $l_1$ space, then the gradient $\nabla f(\vb{x})$ lives in the dual $l_\infty$ space, and hence adding these vectors is not even formally well-defined.

Mirror descent may be viewed as a generalization of gradient descent to non-Euclidean geometries. To gain intuition for why we might want to do this, recall that minimizing a function that is $L$-Lipschitz in the Euclidean norm to precision $\epsilon$ requires $O\qty(\frac{L^2}{\epsilon^2})$ iterations using the above projected gradient descent bound. Note that this expression does not have any explicit dependence on the dimension $p$. However, if the parameter $L$ has a dependence on $p$, then the convergence rate could depend on $p$ implicitly. Consider for example a situation in which we know that all partial derivatives of $f$ are bounded by $1$, so that $\| \nabla f(\vb{x}) \|_\infty \leq 1$ for all $\vb{x}$ in the domain. Then it follows that we can bound $L \leq \sqrt{p}$, and so we obtain an upper bound of $O\qty(\frac{p}{\epsilon^2})$  for gradient descent, which has a linear dependence on $p$. But notice that under this assumption, we have a much stronger bound on the $\infty$-norm of the gradient. In particular, $\| \nabla f\|_\infty \leq 1$, so the Lipschitz constant is only $1$ with respect to this geometry. If we could somehow work in an $l_1$ geometry so that $\| \nabla f(\vb{x}) \|_\infty$ is the relevant quantity instead of $\| \nabla f(\vb{x}) \|_2$, then perhaps we could achieve a stronger upper bound on the convergence rate.  Indeed, this is possible with mirror descent.

In the remainder of this section, we review basic aspects  of mirror descent and its stochastic variant.  We begin by fixing an arbitrary norm $\| \cdot \|$ on $\mathbb{R}^p$, and a compact convex set $\mathcal{X}\subset \mathbb{R}^p$. Recall that the dual norm is defined by $\| \vb{g} \|_* = \sup_{\vb{x}\in \mathbb{R}^p : \| \vb{x} \| \leq 1} \vb{g}^{\top} \vb{x}$.  Let $\mathcal{D}\subset \mathbb{R}^p$ be a convex open set such that $\mathcal{X} \subset \overline{\mathcal{D}}$, where $\overline{\mathcal{D}}$ is the closure of $\mathcal{D}$. Let $\Phi$ be a real-valued function on $\mathcal{D}$. Call $\Phi$ a \emph{mirror map} (sometimes also called a potential function or distance-generating function) if it satisfies the following technical properties: it is differentiable, strictly convex, $\nabla \Phi(\mathcal{D}) = \mathbb{R}^p$, and $\lim_{\vb{x}\rightarrow \partial \mathcal{D}} \| \nabla \Phi (\vb{x}) \| = \infty$. Intuitively, the mirror map $\Phi$ may be thought of as a distance-generating function appropriate to the geometry of the problem. For reference, an appropriate choice of $\Phi$ for a Euclidean geometry is $\Phi(\vb{x}) = \frac{1}{2} \| \vb{x}\|_2^2$, with $\mathcal{D}$ chosen to be $\mathbb{R}^p$. For an $l_1$ geometry where $\mathcal{X}$ is the unit simplex (so points may be interpreted as probability vectors), an appropriate choice of $\Phi$ is the negative entropy, $\Phi(\vb{x}) = \sum_{i=1}^p \vb{x}_i \log \vb{x}_i$, defined on the positive orthant $\mathcal{D} = \mathbb{R}^p_{++}$. 

We may associate to the mirror map $\Phi$ its \emph{Bregman divergence},
\begin{equation*}
D_\Phi (\vb{x},\vb{y}) = \Phi(\vb{x}) - \qty[ \Phi(\vb{y}) + \nabla \Phi(\vb{y})^{\top} (\vb{x}-\vb{y})].
\end{equation*}
The quantity $D_\Phi(\vb{x},\vb{y})$ may be thought of as a distance measure between $\vb{x}$ and $\vb{y}$, generated by $\Phi$. If $\Phi(\vb{x}) = \frac{1}{2}\|\vb{x}\|_2^2$, then $D_\Phi (\vb{x},\vb{y}) = \frac{1}{2}\| \vb{x} - \vb{y} \|_2^2$. If $\Phi(\vb{x}) = \sum_{i=1}^p \vb{x}_i \log \vb{x}_i$, then $D_\Phi (\vb{x},\vb{y}) = D_{\text{KL}}(\vb{x},\vb{y})$, the (generalized) KL divergence between $\vb{x}$ and $\vb{y}$. We now define the notion of a projection onto the feasible set $\mathcal{X}$ with respect to the Bregman divergence $D_\Phi$:
\begin{equation*}
\Pi^{\Phi}_{\mathcal{X}}(\vb{y}) = \argmin_{\vb{x}\in \mathcal{X}\cap \mathcal{D}} D_{\Phi}(\vb{x},\vb{y}).
\end{equation*}

We are now ready to define the mirror descent procedure, with stepsize $\eta$. Let $\vb{x}_1 = \argmin_{\vb{x}\in \mathcal{X}\cap \mathcal{D}} \Phi(\vb{x})$. Then mirror descent is defined by the following iteration. For $t\geq 1$, let $\vb{y}_{t+1}\in \mathcal{D}$  and $\vb{x}_{t+1} \in \mathcal{X}$ be such that
\begin{equation*}
\nabla \Phi(\vb{y}_{t+1}) = \nabla \Phi(\vb{x}_t) - \eta \nabla f(\vb{x}_t)
\end{equation*}
and
\begin{equation*}
\vb{x}_{t+1} \in \Pi_{\mathcal{X}}^\Phi (\vb{y}_{t+1}).
\end{equation*}

In other words, we first move to a ``dual space'' via the mirror map $\Phi$, then do the gradient descent step in the dual space, then move back to the original space again via the mirror map. The resulting point, $\vb{y}_{t+1}$, may lie outside the feasible set $\mathcal{X}$, so we then project back to $\mathcal{X}$ via the Bregman divergence generated by $\Phi$. We also note that a step of mirror descent can be equivalently described in the following proximal picture, which makes the relation to gradient descent clearer.

\begin{equation*}
\vb{x}_{t+1} = \argmin_{\vb{x}\in \mathcal{X}\cap \mathcal{D}}\qty[ f(\vb{x}_t)+  \nabla f(\vb{x}_t)^{\top} (\vb{x}-\vb{x}_t)  + \frac{1}{\eta} D_{\Phi}(\vb{x},\vb{x}_t) ].
\end{equation*}

The following convergence rate can be proven for mirror descent.

\begin{theorem}
If $\Phi$ is $\rho$-strongly convex on $\mathcal{X}\cap \mathcal{D}$ with respect to $\| \cdot \|$, $R^2 := \sup_{\vb{x}\in \mathcal{X}\cap \mathcal{D}} \qty[ \Phi(\vb{x}) - \Phi(\vb{x}_1) ]$, $f$ is convex, and $f$ is $L$-Lipschitz with respect to $\| \cdot \|$, then mirror descent with $\eta = \frac{R}{L}\sqrt{\frac{2}{T}}$ satisfies
\begin{equation*}
f\qty(\frac{1}{T}\sum_{s=1}^{T} \vb{x}_s) - f(\vb{x}^*) \leq R L \sqrt{\frac{2}{\rho T}}.
\end{equation*}
\end{theorem}

We now record a mirror map $\Phi$ that is appropriate for an $l_1$ setup, extracted from \cite{nemirovski2009robust}. Namely, assuming $\mathcal{X} \subset \mathbb{R}^p$, take
\begin{equation*}
\Phi(\vb{x}) = (e \ln p) \sum_{i=1}^p | \vb{x}_i |^{1+\frac{1}{\ln p}}, \, \, p \geq 3.
\end{equation*}
Whenever $\mathcal{X}$ is contained in an $l_1$-ball of radius $1$ centered at the origin, we have $R^2 = e \ln p$ and $\rho \geq 1$. These assumptions on $\mathcal{X}$ can always be achieved by shifting and scaling $\mathcal{X}$. We now state a mirror descent bound for an $l_1$ geometry.

\begin{theorem}
If the convex function $f$ is $L$-Lipschitz with respect to the norm $\| \cdot \|_1$, and $\mathcal{X}$ is contained in a 1-ball of radius $R_1$, then projected mirror descent with an appropriate choice of mirror map and stepsizes satisfies
\begin{equation*}
f\qty(\frac{1}{T}\sum_{s=1}^{T} \vb{x}_s) - f(\vb{x}^*) \leq R_1 L \sqrt{\frac{2e\ln p}{T}}
\end{equation*}
\end{theorem}

Finally, we comment on stochastic mirror descent (specializing to the $l_1$ setup case). In the stochastic setting, one is given access to a oracle which, upon input $\vb{x} \in \mathcal{X}$, outputs a random variable $\hat{\vb{g}}(\vb{x})$ such that $\E \hat{\vb{g}}(\vb{x}) = \nabla f(\vb{x})$ and $\E \| \hat{\vb{g}}(\vb{x}) \|_\infty^2 \leq G_\infty^2$. The iteration for stochastic mirror descent is identical to that of mirror descent, except the gradients are replaced by their stochastic estimates, just as for the case of stochastic gradient descent. As for the Euclidean case of SGD, the upper bound obtained for SMD is qualitatively very similar to that of the noiseless version. 

\begin{theorem}
Assume the convex function $f$ is  contained on a 1-ball of radius $R_1$. Assume we have access to a stochastic gradient oracle, which upon input of $\vb{x}\in \mathcal{X}$, returns a random vector $\hat{\vb{g}}(\vb{x})$ such that $\mathbb{E} \hat{\vb{g}}(\vb{x}) = \nabla f(\vb{x})$ and $\mathbb{E} \| \hat{\vb{g}}(\vb{x}) \|_\infty^2 \leq G_\infty^2$. Then SMD with appropriate stepsize  satisfies 
\begin{equation*}
f\qty(\frac{1}{T}\sum_{s=1}^{T} \vb{x}_s) - f(\vb{x}^*) \leq R_1 G_\infty \sqrt{\frac{2e\ln p}{T}}
\end{equation*}
\end{theorem}

Finally, if $f$ is strongly convex with respect to the norm $\| \cdot \|_1$, then SMD can be accelerated similarly to how SGD can be accelerated for strongly convex functions. See for example \cite{hazan2014beyond}.

\begin{theorem}
Let $f$ be $\lambda_1$-strongly convex on $\mathcal{X}$ with respect to norm $\|\cdot \|_1$. Assume we have access to a stochastic gradient oracle, which upon input of $\vb{x}\in \mathcal{X}$, returns a random vector $\hat{\vb{g}}(\vb{x})$ such that $\mathbb{E} \hat{\vb{g}}(\vb{x}) = \nabla f(\vb{x})$ and $\mathbb{E} \| \hat{\vb{g}}(\vb{x}) \|_\infty^2 \leq G_\infty^2$. Then a SMD-type algorithm outputs a vector $\bar{\vb{x}}$ for which 
\begin{equation*}
\E f(\bar{\vb{x}}) - f(\vb{x}^*)  \leq \frac{16 G_\infty^2}{\lambda_1 T} .
\end{equation*}
\end{theorem}

\end{document}